\documentclass[5p]{elsarticle}

\usepackage{lineno}
\usepackage[colorlinks=false]{hyperref}
\usepackage{adjustbox}
\usepackage{graphicx}
\usepackage{amsmath,amsthm,amssymb}
\usepackage{mathabx}
\usepackage{algorithmic}
\usepackage{algorithm}
\usepackage{xcolor,soul}
\usepackage{empheq}
\usepackage{csquotes}
\usepackage{tabularx}
\usepackage{float}
\usepackage{subfig}
\usepackage{cancel}
\usepackage[english]{babel}
\usepackage{booktabs}
\usepackage{stackengine}
\usepackage{afterpage}
\usepackage{accents}
\usepackage{enumerate}
\usepackage{setspace}
\usepackage{url} 
\usepackage{mathtools}
\usepackage{bm}
\usepackage{times}
\usepackage{multirow}

\modulolinenumbers[5]
\journal{Elsevier Systems \& Control Letters. Published article \cite{BONASSI2021105049} available at \href{https://doi.org/10.1016/j.sysconle.2021.105049}{doi.org/10.1016/j.sysconle.2021.105049}}
\newtheorem{assumption}{Assumption}
\newtheorem{theorem}{Theorem}
\newtheorem{lemma}{Lemma}
\newtheorem{proposition}{Proposition}
\newtheorem{remark}{Remark}
\newtheorem{definition}{Definition}
\newtheorem{corollary}{Corollary}

\newcommand{\ubar}[1]{\underaccent{\bar}{#1}}
\newcommand{\sss}[1]{_{\scriptscriptstyle #1}}
\newcommand{\uss}[1]{^{\scriptscriptstyle #1}}

\hypersetup{
	colorlinks,
	citecolor=black,
	filecolor=black,
	linkcolor=black,
	urlcolor=black
}

\begin{document}

\begin{frontmatter}
\title{On the stability properties of Gated Recurrent Units neural networks\tnoteref{copyright}\tnoteref{tfunds}}

\tnotetext[copyright]{© 2021. This manuscript version is made available under the CC-BY-NC-ND 4.0.}
\tnotetext[tfunds]{This research did not receive any specific grant from funding agencies in the public, commercial, or not-for-profit sectors.}

\author[1]{Fabio Bonassi\corref{corr1}} 
\ead{fabio.bonassi@polimi.it}
	
\author[1]{Marcello Farina}

\author[1]{Riccardo Scattolini}

\address[1]{Dipartimento di Elettronica, Informatica e Bioingegneria, Politecnico di Milano, 20133 Milano, Italy}
\cortext[corr1]{Corresponding author}

\begin{abstract}
The goal of this paper is to provide sufficient conditions for guaranteeing the Input-to-State Stability (ISS) and the Incremental Input-to-State Stability ($\delta$ISS) of Gated Recurrent Units (GRUs) neural networks.
These conditions, devised for both single-layer and multi-layer architectures, consist of nonlinear inequalities on network's weights. They can be employed to check the stability of trained networks, or can be enforced as constraints during the training procedure of a GRU.
The resulting training procedure is tested on a Quadruple Tank nonlinear benchmark system, showing remarkable modeling performances.
\end{abstract} 


\end{frontmatter}

\section{Introduction}
Neural Networks (NNs) have gathered increasing attention from the control systems community.
The approximation capabilities of NNs \cite{hornik1989multilayer, schafer2006recurrent}, the spread of reliable tools to train, test and deploy them, as well as the availability of large amounts of data, collected from the plants in various operating conditions, fostered the adoption of NNs in data-driven control applications. 

A standard approach is to use a NN to identify a dynamical system and then, based on such model, design a regulator by means of traditional model-based control strategies.
In particular, Model Predictive Control (MPC) can be adopted in combination with such models, as it allows to cope with nonlinear system models and it can guarantee the closed-loop stability, even in presence of input constraints.
For identification purposes, Feed-Forward Neural Networks (FFNNs) were initially adopted  \cite{hunt1992neural, levin1993control}, thanks to their simple structure and easy training.
FFNNs have been soon abandoned due to their structural lack of memory, which prevents them from achieving accurate long-term predictions.
To overcome these limits, Recurrent Neural Networks have been introduced: in particular, among the wide variety of recurrent architectures, very promising ones for system identification are Long-Short Term Memory networks (LSTMs, \cite{hochreiter1997long}), Echo State Networks (ESNs, \cite{jaeger2002tutorial}), and Gated Recurrent Units (GRUs, \cite{cho2014learning}), see \cite{mohajerin2019multistep, rehmer2019using, wu2019machine, wu2021machine,  ogunmolu2016nonlinear}.

% While they found great practical application in the industry and accademy
Owing to their remarkable modeling performances, these architectures enjoy broad applicability, for example in chemical \cite{atuonwu2010identification} and pharmaceutical process control \cite{wong2018recurrent}, manufacturing plants management \cite{lanzetti2019recurrent}, buildings' HVAC optimization \cite{terzi2020learning}.
However, despite their popularity among the practitioners, only little theoretical results are available on Recurrent NNs.
In \cite{stipanovic2018someoriginal} and \cite{deka2019global}, the stability of autonomous LSTMs is studied, but they do not account for manipulable inputs. 
Similarly, in \cite{stipanovic2020stability} stability considerations have been recently carried out for autonomous GRUs.
Miller and Hardt, in \cite{miller2018stable}, provided sufficient conditions for the  stability of recurrent networks, stated as inequalities on network's parameters.

For some recurrent architectures, more advanced stability properties have been recently studied, namely the Input-to-State Stability (ISS) \cite{jiang2001input} and the Incremental Input-to-State Stability ($\delta$ISS) \cite{bayer2013discrete}.
The ISS property guarantees that, regardless of the initial conditions, bounded inputs or disturbances lead to bounded network's states. 
In \cite{bonassi2019lstm} the authors derived a sufficient condition under which LSTMs are guaranteed to be ISS, inferring the boundedness of the network's output reachable set. 
This boundedness has thus been leveraged to perform a probabilistic safety verification of the network.
The $\delta$ISS property is a {stronger property than the ISS and implies that, asymptotically, the closer are two input sequences applied to the network, the smaller is the bound on the maximum distance between the resulting state trajectories.}
In \cite{terzi2021learning} the authors provided sufficient conditions for the $\delta$ISS of LSTMs, exploiting this property to design a converging state observer and a stabilizing MPC control law.
Both stability properties are also useful, among other applications, for Robust MPC  \cite{bayer2013discrete} and Moving Horizon Estimators \cite{alessandri2008moving} design.
Analogous stability conditions have been retrieved for ESNs in \cite{armenio2019model}.

To the best of authors' knowledge, no theoretical result is currently available concerning the ISS and $\delta$ISS of GRUs.
We believe that this gap needs to be filled, as GRUs -- although simpler -- achieve comparable, or even superior, results with respect to LSTMs when it comes to modeling dynamical systems \cite{rehmer2019using, ogunmolu2016nonlinear, bianchi2017overview}.

The purpose of this paper is twofold. 
First GRUs are recast in state-space form, and sufficient conditions for their ISS and $\delta$ISS are retrieved, both for single-layer and deep (i.e. multi-layer) networks.
These conditions come in the form of nonlinear inequalities on network's weights, and  can be employed to certify the stability of a trained network, or can be enforced as constraints during the training procedure to guarantee the stability of the GRU.
{Thus, when the system to be learned enjoys the ISS or $\delta$ISS property, it is possible to ensure the consistency of the GRU model to the actual system. 
The model's stability can then be leveraged during the controller design phase.}
Secondly, this approach is tested on a Quadruple Tank nonlinear benchmark system \cite{alvarado2011comparative}, showing satisfactory performances.
Guidelines are also provided about how the training procedure of these stable GRUs can be carried out in a common environment, TensorFlow, which does not support  constrained training.

This paper is organized as follows. In Section \ref{sec:model} the state-space model of GRUs is formulated, and the existence of an invariant set for network's states is shown.
The ISS and $\delta$ISS properties of these networks are then studied in Section \ref{sec:single}, and the results are extended to deep GRUs in Section \ref{sec:deep}.
In Section \ref{sec:example} the proposed method is tested on the Quadruple Tank benchmark system.

\medskip
\noindent \textbf{Notation and preliminaries}
In the paper we adopt the following notation. Given a vector $v$, we denote by $v^{\prime}$ its transpose, by $\| v \|$ its Euclidean norm and by $\| v \|_\infty$ its infinity-norm. 
The $j$-th component of $v$ is indicated by $v_j$, and its absolute value by $\lvert v_j \lvert$. 
Boldface indicates a sequence of vectors, i.e. $\bm{v} = \{ v(0), v(1), ...\}$, where $\| \bm{v} \|_p = \max_{k \geq 0} \| v(k) \|_p$. 
If $v_a$ and $v_b$ are two distinct vectors,  $v_{\{a, b\}} \in \mathcal{V}$ is used to indicate that  $v_a \in \mathcal{V}$ and $v_b \in \mathcal{V}$. By extension, an inequality containing $v_{\{a, b\}}$ is intended to hold both by $v_a$ and $v_b$.
For multi-layer networks the superscript $i$, e.g. $v^i$, denotes a quantity referred to the $i$-th layer.
For conciseness, the discrete-time instant $k$ may be dropped in no ambiguity occurs, and $v^+$ may be used to denote the value of vector $v$ at time $k+1$.
The Hadamard (i.e. element-wise) product between $u$ and $v$ is indicated by $u \circ v$.
The sigmoid and hyperbolic tangent activation functions are respectively denoted by $\sigma(x) = \frac{1}{1 + e^{-x}}$ and $\phi(x) = \text{tanh}(x)$.
If the argument of $\phi$ and $\sigma$ is a vector, these activation functions are intended to be applied element-wise.
{For a given $n \times m$ matrix $A$, we denote by $\| A \|_\infty$ its induced $\infty$-norm, which is defined as
\begin{equation*}
    \| A \|_\infty = \max_v \frac{\| A v \|_\infty}{\|v \|_\infty} = \max_{1 \leq i \leq n} \sum_{j=1}^m \lvert a_{ij} \lvert,
\end{equation*}
where $a_{ij}$ denotes the element of $A$ in position $i, j$. 
The infinity norm satisfies the homogeneity condition, i.e. for any scalar $\alpha$ it holds that $\| \alpha A \|_\infty = \lvert \alpha \lvert \| A \|_\infty$. 
Moreover, given a matrix $B$ of suitable shape, it holds that $\| A + B \|_\infty \leq \| A \|_\infty + \| B \|_\infty$, and $\| A \cdot B \|_\infty \leq \| A \|_\infty \, \| B \|_\infty$.
Finally, for any vector $v \in \mathbb{R}^m$, $\| A v \|_\infty \leq \| A \|_\infty \, \| v \|_\infty$.}

\section{Single-layer GRU model} \label{sec:model}
Let us consider the following neural network, obtained combining a single GRU layer, as defined in \cite{cho2014learning}, and a linear output transformation
\begin{equation} \label{eq:model:gru}
	\left\{ \begin{array}{l}
		x^+ = z\circ x + (1 - z) \circ \phi\left( W_r \, u + U_r \, f \circ x + b_r \right) \\
		z = \sigma \left( W_z \, u + U_z \, x + b_z \right) \\
		f = \sigma \left( W_f \, u + U_f \, x + b_f \right) \\
		y = U_o \, x + b_o
	\end{array}\right.,
\end{equation}
where $x \in \mathbb{R}^{n_x}$ is the state vector, $u \in \mathbb{R}^{n_u}$ is the input vector, $y \in \mathbb{R}^{n_o}$ is the output vector. 
Moreover, $z = z(u, x)$ is called \emph{update} gate, and $f = f(u, x)$ is known as \emph{forget} gate. 
The matrices $W_{\star}$, $U_\star$, and $b_\star$, are the weights and biases that parametrize the model.
% which shall be properly tuned during the training procedure by minimizing a prescribed loss function.
%TODO: CITARE fonte
\begin{assumption} \label{ass:u_bounded}
	The input $u$ is unity-bounded
	\begin{equation} \label{eq:model:u_bound}
		u \in \mathcal{U} \subseteq \left[-1, 1\right]^{n_u},
	\end{equation}
	i.e. $\| u \|_\infty \leq 1$.
\end{assumption}
Note that this assumption is quite customary when dealing with neural networks, see e.g. \cite{goodfellow2016deep}, and can be easily satisfied by means of a suitable normalization of the input vector.
Before stating the first, instrumental, theoretical results of this paper, let us remind that $\sigma$ and $\phi$ are bounded as follows
\begin{subequations} \label{eq:model:sigma_bounds}
	\begin{align}
		\sigma(\cdot) & \in (0, 1), \\
		\phi(\cdot) & \in (-1, 1),
	\end{align}
\end{subequations}
and that they are Lipschitz-continuous with Lipschitz coefficients $L_\sigma = \frac{1}{4}$ and $L_\phi = 1$, respectively \cite{miller2018stable}.
Hereafter, for the sake of compactness, the following notation may be used
\begin{equation} \label{eq:model:r_def}
	r(u, x) = \phi\left( W_r \, u + U_r \, f \circ x + b_r \right).
\end{equation}
The following preliminary results are instrumental for the reminder of the paper.
\begin{lemma} \label{lemma:single:invset_ini}
	$\mathcal{X} = [-1, 1]^{n_x}$ is an invariant set of the state $x$ of system \eqref{eq:model:gru}, i.e. for any input $u$
	$$ x(k) \in \mathcal{X} \, \Rightarrow \, x(k+1) \in \mathcal{X}. $$ 
\end{lemma}
\begin{proof}
	Consider the $j$-th	state, and let $\omega_j(k) = z(u(k), x(k))_j$, and $\eta_j(k) = r(u(k), x(k))_j$. 
	It is possible to recast \eqref{eq:model:gru} as a Linear Parameter Varying (LPV) system
	\begin{equation} \label{eq:model:proof_invset_ini:ltv}
		x_j(k+1) = \omega_j(k) \, x_j(k) + \big(1 - \omega_j(k) \big) \, \eta_j(k).
	\end{equation}
	Then, in light of \eqref{eq:model:sigma_bounds}, it holds that $\omega_j(k) \in (0, 1)$ and $\eta_j(k) \in (-1, 1)$.
	Since $x_j(k+1)$ is a convex combination of two quantities  bounded in $[-1, 1]$, it follows that $x_j(k+1) \in [-1, 1]$.
	Thus, $x(k+1) \in \mathcal{X}$.
\end{proof}

\begin{lemma} \label{lemma:single:invset}
	For any arbitrary initial state $\bar{x} \in \mathcal{R}^{n_x}$,
	\begin{enumerate}[i.]
	    \itemsep0em
	    \item if $\bar{x} \notin \mathcal{X}$, $\| x(k) \|_\infty$ is strictly decreasing until $x(k) \in \mathcal{X}$;
	    \item the convergence happens in finite time, i.e. there exists a finite $\bar{k} \geq 0$ such that $x(k) \in \mathcal{X}, \, \forall k \geq \bar{k}$;
	    \item each state component $x_j$ converges into its invariant set $[-1, 1]$ in an exponential fashion.
	\end{enumerate}
\end{lemma}
\begin{proof}
    See \ref{appendix:single}.
\end{proof}

In the remainder, the following assumption is taken.
\begin{assumption} \label{ass:single:initial_state}
    The initial state of the GRU network \eqref{eq:model:gru} belongs to an arbitrarily large, but bounded, set $ {\check{\mathcal{X}}}\supseteq \mathcal{X}$, defined as
    \begin{equation}
        {\check{\mathcal{X}}} = \{ x \in \mathbb{R}^{n_x}: \, \| x \|_\infty \leq {\widecheck{\lambda}} \},
    \end{equation}
    with ${{\widecheck{\lambda}}} \geq 1$.
\end{assumption}

\section{Stability properties of single-layer GRUs} \label{sec:single}
The  goal of this section is to provide sufficient conditions for the ISS and $\delta$ISS of single-layer GRUs in the form of \eqref{eq:model:gru}.
The results will be later extended to multi-layer networks.
For compactness, in the following we denote by $x(k, \bar{x}, \bm{u}, b_r)$ the state at time $k$ of the system \eqref{eq:model:gru}, fed by the sequence $\bm{u} = \{u(0), u(1), ... \lvert u(t) \in \mathcal{U} \}$, and characterized by the initial state $x(0) = \bar{x} \in \check{\mathcal{X}}$.
Recalling from \cite{jiang2001input} the definitions of $\mathcal{K}_\infty$ and $\mathcal{KL}$ functions, the following definition of ISS is given.

\begin{definition}[ISS] \label{def:ISS}
	System \eqref{eq:model:gru} is Input-to-State Stable if there exist functions $\beta \in \mathcal{KL}$, $\gamma_u \in \mathcal{K}_\infty$, and $\gamma_b \in \mathcal{K}_\infty$, such that for any $k \in \mathbb{Z}_{\geq 0}$, any initial condition $\bar{x} \in \check{\mathcal{X}}$, any value of $b_r$, and any input sequence $\bm{u}$, it holds that
	\begin{equation} \label{eq:def:ISS}
		\| x(k, \bar{x}, \bm{u}, b_r) \|_\infty \leq \beta(\| \bar{x} \|_\infty, k) +\gamma_u(\| \bm{u} \|_\infty )  + \gamma_b(\| b_r \|_\infty ).
	\end{equation}
\end{definition}

\begin{remark} \label{rmk:hinf}
Differently from \cite{jiang2001input}, Definition \ref{def:ISS} features the in\-finity-norm of the state vector.
It is possible to show that this definition implies the one given by Jiang et al \cite{jiang2001input}. Indeed, recalling that
{$ \frac{1}{\sqrt{n}} \| v \| \leq \| v \|_\infty \leq \| v \|$}, and that $\beta$ is monotonically increasing in its first argument, \eqref{eq:def:ISS} can be recast as
	\begin{equation*}
		\| x(k, \bar{x}, \bm{u}, b_r) \| \leq \sqrt{n_x} \Big[ \beta(\| \bar{x} \|, k) + \gamma_u(\| \bm{u} \| ) + \gamma_b(\| b_r \| ) \Big].
	\end{equation*}
\end{remark}

\begin{theorem} \label{thm:single:iss}
	A sufficient condition for the ISS of the single-layer GRU network \eqref{eq:model:gru} is that
	\begin{equation} \label{eq:single:iss:condition}
		\| U_r \|_\infty \, \bar{\sigma}_f < 1,
	\end{equation}
	where
	\begin{equation} \label{eq:single:iss:bounds}
		\bar{\sigma}_f = \sigma \left( \| W_f \quad U_f \quad b_f \|_\infty \right).
	\end{equation}
\end{theorem}
\begin{proof}
    See \ref{appendix:single}.
\end{proof}

As mentioned above, ISS represents a fundamental property for the model of a dynamical system, as it also allows to retrieve a bound for the states around the origin, which can be seen as a (conservative) estimation of the model's output reachable set \cite{bonassi2019lstm}.
However, especially in the realm of robust control, a further property is desirable, i.e. the $\delta$ISS \cite{bayer2013discrete}. 
This property is here stated using the infinity-norm of the state vector. Nonetheless,  as discussed in Remark \ref{rmk:hinf},  this formulation implies the one provided by Bayer et al. \cite{bayer2013discrete}.

\begin{definition}[$\delta$ISS] \label{def:deltaISS}
    	System \eqref{eq:model:gru} is Incrementally Input-to-State Stable ($\delta$ISS) if there exist functions $\beta\sss{\Delta} \in \mathcal{KL}$ and $\gamma_{\Delta u} \in \mathcal{K}_\infty$ such that, for any $k \in \mathbb{Z}_{\geq 0}$, any pair of initial states $\bar{x}_{a} \in {\check{\mathcal{X}}}$ and $\bar{x}_{b} \in {\check{\mathcal{X}}}$, and any pair of input sequences $\bm{u}_a $ and $\bm{u}_b$, it holds that
	\begin{equation}\label{eq:def:deltaiss}
	\begin{aligned}
		&\| x(k, \bar{x}_{a}, \bm{u}_a, b_r) - x(k, \bar{x}_{b}, \bm{u}_b, b_r) \|_\infty  \\
		& \qquad \leq  \beta_{\Delta}( \| \bar{x}_{a} - \bar{x}_{b} \|_\infty, k) + \gamma_{{\Delta} u}(\| \bm{u}_a - \bm{u}_b \|_\infty)
	\end{aligned}
	\end{equation}
\end{definition}

This property implies that, initializing a $\delta$ISS network with different initial conditions ($\bar{x}_a$ and $\bar{x}_b$), and feeding it with different input sequences ($\bm{u}_a$ and $\bm{u}_b$), one obtains state trajectories whose maximum distance is asymptotically bounded by a function  monotonically increasing with the maximum distance between the input sequences.
It is worth noticing that, among other things, this property ensures that when GRUs are used to model nonlinear dynamical systems, their performances are not biased by a wrong initialization of the network, since $\beta(\cdot, k) \to 0$ as $k \to \infty$.

In the following, a condition ensuring that the network is $\delta$ISS is hence provided.

\begin{theorem} \label{thm:single:deltaiss}
	A sufficient condition for the $\delta$ISS of the single-layer GRU network \eqref{eq:model:gru} is that
	\begin{equation} \label{eq:single:deltaiss:condition}
		\| U_r \|_\infty \left( \frac{1}{4} {\widecheck{\lambda}} \| U_f \|_\infty +{\widecheck{\sigma}}_f \right) < 1 - \frac{1}{4} \frac{{\widecheck{\lambda}} +{\widecheck{\phi}}_r}{1 -{\widecheck{\sigma}}_z} \| U_z \|_\infty,
	\end{equation}
	where
	\begin{subequations} \label{eq:single:deltaiss:bounds}
	\begin{align}
		 {\widecheck{\sigma}}_z &= \sigma( \| W_z \quad {\widecheck{\lambda}} U_z \quad b_z \|_\infty ), \label{eq:single:deltaiss:bounds:z}\\
		 { \widecheck{\sigma}}_f &= \sigma( \| W_f \quad {\widecheck{\lambda}} U_f  \quad b_f \|_\infty ), \label{eq:single:deltaiss:bounds:f} \\
		 {\widecheck{\phi}}_r &= \phi( \| W_r \quad {\widecheck{\lambda}} U_r \quad b_r \|_\infty ). \label{eq:single:deltaiss:bounds:r}
	\end{align}
	\end{subequations}
\end{theorem}

\begin{proof}
	See \ref{appendix:single}.
\end{proof}

It should be noted that the conditions stated in Theorem \ref{thm:single:deltaiss} might be very conservative.
To relax the conservativeness of the approach, one can assume that the GRU network is always initialized inside the invariant set, i.e. $ {\check{\mathcal{X}}}= \mathcal{X}$, which allows to ease bounds \eqref{eq:single:deltaiss:bounds} and to relax condition \eqref{eq:single:deltaiss:condition}, as shown in the following Corollary.

\begin{corollary} \label{cor:single:deltaiss_relaxed}
A sufficient condition for the $\delta$ISS of the single-layer GRU network \eqref{eq:model:gru}, initialized within $ {\check{\mathcal{X}}}= \mathcal{X}$, is that
	\begin{equation} \label{eq:single:deltaiss:condition_relaxed}
		\| U_r \|_\infty \left( \frac{1}{4}  \| U_f \|_\infty + \bar{\sigma}_f \right) < 1 - \frac{1}{4} \frac{1 + \bar{\phi}_r}{1 - \bar{\sigma}_z} \| U_z \|_\infty,
	\end{equation}
	where
	\begin{subequations} \label{eq:single:deltaiss:bounds_relaxed}
	\begin{align}
		\bar{\sigma}_z &= \sigma( \| W_z \quad U_z \quad b_z \|_\infty ),\\
		\bar{\sigma}_f &= \sigma( \| W_f \quad U_f  \quad b_f \|_\infty ), \\
		\bar{\phi}_r &= \phi( \| W_r \quad U_r \quad b_r \|_\infty ).
	\end{align}
	\end{subequations}
\end{corollary}
\begin{proof}
    Since $ {\check{\mathcal{X}}}= \mathcal{X}$, it follows that ${\widecheck{\lambda}} = 1$.
    Corollary \ref{cor:single:deltaiss_relaxed} can be hence proven invoking Theorem \ref{thm:single:deltaiss} with ${\widecheck{\lambda}} = 1$.
\end{proof}

\begin{remark} \label{rmk:weaker_deltaiss_single}
    Note that the condition \eqref{eq:single:deltaiss:condition_relaxed} involved by Corollary \ref{cor:single:deltaiss_relaxed} is less conservative than the condition \eqref{eq:single:deltaiss:condition} required by Theorem \ref{thm:single:deltaiss}.
    While Corollary \ref{cor:single:deltaiss_relaxed} ensures the $\delta$ISS just inside the invariant set $\mathcal{X}$, it allows to guarantee a similar but weaker stability-related property also when $\bar{x}_a \notin \mathcal{X}$ and/or $\bar{x}_b \notin \mathcal{X}$.
    In this regard, it is not possible to show that, during the exponential convergence of the states $x_a$ and $x_b$ into $\mathcal{X}$ (Lemma \ref{lemma:single:invset}), the $\delta$ISS relation \eqref{eq:def:deltaiss} is implied by condition \eqref{eq:single:deltaiss:condition_relaxed}.
    However, as soon as $x_{\{a,b\}} \in \mathcal{X}$ -- which is guaranteed to happen in finite time (Lemma \ref{lemma:single:invset}) -- the $\delta$ISS property regularly applies.
\end{remark}

\begin{remark} \label{rmk:constraints}
	Theorem \ref{thm:single:iss}, Theorem \ref{thm:single:deltaiss}, and Corollary \ref{cor:single:deltaiss_relaxed} involve constraints on the infinity-norms of the weight matrices. 
	These conditions can be used to a-posteriori check if the trained network is ISS and $\delta$ISS, or they can be used to enforce these stability properties during training.  
	In the latter case, {since the main available training environment are unconstrained, conditions  \eqref{eq:single:iss:condition}, \eqref{eq:single:deltaiss:condition}, and \eqref{eq:single:deltaiss:condition_relaxed} may be implemented as soft constraints in the training procedure by penalizing their violation in the loss function, as discussed in Section \ref{sec:example}. 
	In this way, nonlinear constraints are recast as nonlinear additive terms on the cost function, which can be easily managed by gradient-based training algorithms.}
\end{remark}

Finally, it is worth noticing that the $\delta$ISS property implies ISS \cite{bayer2013discrete}. 
Our sufficient conditions are consistent with this relation, as shown by the following Proposition.

\begin{proposition} \label{prop:iss_deltaiss}
	If the GRU network \eqref{eq:model:gru} satisfies the $\delta$ISS condition \eqref{eq:single:deltaiss:condition}, it also satisfies the ISS condition \eqref{eq:single:iss:condition}.
\end{proposition}
\begin{proof}
	Since $ {\widecheck{\sigma}}_r \in (0, 1)$, $ {\widecheck{\sigma}}_z \in (0, 1)$, and $\| U_z \|_\infty \geq 0$, {being ${\widecheck{\lambda}}$ strictly positive,} if \eqref{eq:single:deltaiss:condition} is fulfilled it also holds that
	\begin{equation} \label{eq:single:deltaiss_imples_iss:int1}
		\| U_r \|_\infty \bigg( \frac{1}{4} {\widecheck{\lambda}} \| U_f \|_\infty +{\widecheck{\sigma}}_f \bigg) < 1 - \frac{1}{4} \frac{{\widecheck{\lambda}} +{\widecheck{\phi}}_r}{1 -{\widecheck{\sigma}}_z} \| U_z \|_\infty < 1.
	\end{equation}
	Noting that $\frac{1}{4} {\widecheck{\lambda}} \| U_f \|_\infty \geq 0$,  \eqref{eq:single:deltaiss_imples_iss:int1} entails that
	\begin{equation} \label{eq:single:deltaiss_implies:iss:int2}
		\| U_r \|_\infty \,{\widecheck{\sigma}}_f < 1.
	\end{equation}
	Thanks to the monotonicity of $\sigma$, since ${\widecheck{\lambda}} \geq 1$, it holds that $\bar{\sigma}_f \leq{\widecheck{\sigma}}_f$, and hence \eqref{eq:single:deltaiss_implies:iss:int2} implies the ISS condition \eqref{eq:single:iss:condition}.
\end{proof}

\section{Stability properties of deep GRUs}\label{sec:deep}
Despite in many cases single-layer GRUs may show satisfactory performances, in the literature deep (i.e. multi-layer) GRUs are typically adopted to enhance the representational capabilities of these networks \cite{bianchi2017overview, goodfellow2016deep}. Let the superscript $^i$ indicate the $i$-th layer of the network. A deep GRU with $M$ layers is then described by the following equations
\begin{subequations} \label{eq:model:deepgru}
\begin{equation} \label{eq:model:deepgru:gru}
\left\{ \begin{array}{l}
    x^{{i},+} = z^{i}\circ x^{i} + (1 - z^{i}) \circ \phi\big( W_r^{i} \, u^{i} + U_r^{i} \, f^{i} \circ z^{i} + b_r^{i} \big) \vspace{1mm}\\
    z^{i} = \sigma \big( W_z^{i} \, u^{i} + U_z^{i} \, x^{i} + b_z^{i} \big) \vspace{1mm} \\
    f^{i} = \sigma \big( W_f^{i} \, u^{i} + U_f^{i} \, x^{i} + b_f^{i} \big)
\end{array}\right.\!\!,
\end{equation}
for all $i \in \{ 1, ..., M\}$. The input of each layer is the future state of the previous one, save for the first layer which is fed by $u$, i.e.
\begin{equation}\label{eq:model:deepgru:input}
\left\{
\begin{array}{l}
	u^{1} = u,\vspace{1mm} \\
	u^{i} = x^{i-1, +} \qquad \forall i \in \{2, ..., M\}, \\
\end{array} \right.
\end{equation}
while the output of the network is a linear combination of the states of the last layer
\begin{equation} \label{eq:model:deepgru:output}
	y = U_o \, x^{M} + b_o.
\end{equation}
\end{subequations}
The deep GRU described by \eqref{eq:model:deepgru} is depicted in Figure \ref{fig:deep_gru}.
Note that for this network the state is given by the concatenation of all layers' states, i.e. $x = [x^{1 \prime}, ..., x^{M \prime}]^{\prime}$. Similarly we define $b_r = [b_r^{1 \prime}, ..., b_r^{M \prime}]^{\prime}$. 
Lemma \ref{lemma:single:invset_ini} and \ref{lemma:single:invset} are now extended to deep GRUs.

\begin{figure}[t]
    \centering
    \includegraphics[width=1\columnwidth]{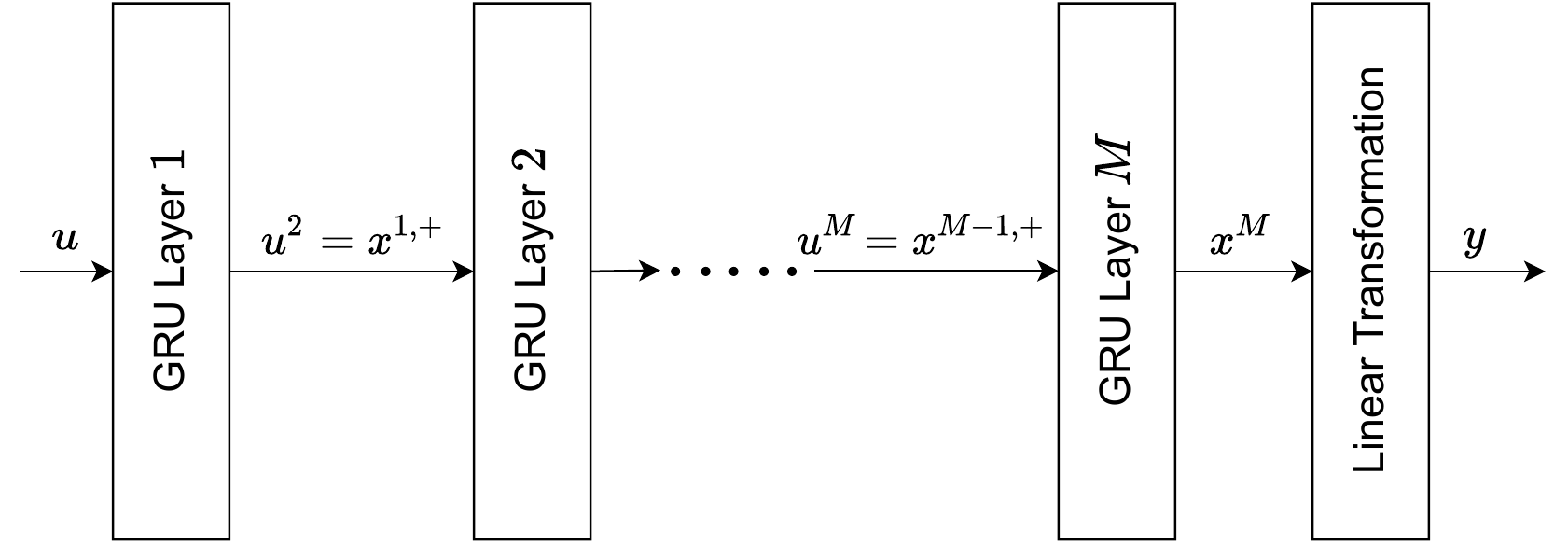}
    \caption{Scheme of the deep GRU \eqref{eq:model:deepgru}.}
    \label{fig:deep_gru}
\end{figure}

\begin{lemma}
    The set $\mathcal{X} = \bigtimes_{i=1}^M \mathcal{X}^i$, with $\mathcal{X}^i = [-1, 1]^{n_x^i}$, is an invariant set of the state $x$ of the deep GRU \eqref{eq:model:deepgru}, meaning that for any input $u$
    \begin{equation*}
        x(k) \in \mathcal{X} \Rightarrow x(k+1) \in \mathcal{X}.
    \end{equation*}
\end{lemma}
\begin{proof}
{Consider the generic layer $i \in \{ 1, ..., M \}$. 
    Since $x^i(k) \in \mathcal{X}^i$, Lemma \ref{lemma:single:invset_ini} implies that, for any $u^i$, $x^i(k+1) \in \mathcal{X}^i$. 
    Hence, $\mathcal{X}^i=[-1, 1]^{n_x^i}$ is an invariant set of the $i$-th layer's state, $x^i$.}
    The Cartesian product of these sets, i.e. $\mathcal{X}$, is therefore the invariant set of the state vector $x$.
\end{proof}

\begin{lemma} \label{lemma:deep:invset}
	For any arbitrary initial state $\bar{x}= [ \bar{x}^{1 \prime}, ..., \bar{x}^{M \prime}]^\prime$, with $\bar{x}^i \in \mathbb{R}^{n_x^i}$,
	\begin{enumerate}[i.]
	    \itemsep0em
	    \item if $\bar{x}^i \notin \mathcal{X}^i$, $\| x^i(k) \|_\infty$ is strictly decreasing until $x^i(k) \in \mathcal{X}^i$, and $\| x(k) \|_\infty$ is strictly decreasing until $x(k) \in \mathcal{X}$;
	    \item the convergence happens in finite time, i.e. there exists a finite $\bar{k} \geq 0$ such that for any layer $x^i(k) \in \mathcal{X}^i, \, \forall k \geq \bar{k}$;
	    \item each state component $x_j^i(k)$ converges into its invariant set $[-1, 1]$ in an exponential fashion.
	\end{enumerate}
\end{lemma}
\begin{proof}
    See \ref{appendix:deep}.
\end{proof}

As for single-layer GRUs, in the remainder the following assumption is taken.
\begin{assumption} \label{ass:deep:initial_state}
    The initial state of the GRU network \eqref{eq:model:deepgru:gru} belongs to an arbitrarily large, but bounded, set $ {\check{\mathcal{X}}}\supseteq \mathcal{X}$, defined as $ {\check{\mathcal{X}}}= \bigtimes_{i=1}^M {\check{\mathcal{X}}}^i$, where
    \begin{equation}
        {\check{\mathcal{X}}}^i = \{ x \in \mathbb{R}^{n_x^i}: \, \| x \|_\infty \leq {\widecheck{\lambda}}^i \},
    \end{equation}
    with ${\widecheck{\lambda}}^i \geq 1$.
\end{assumption}

The following propositions can hence be stated to guarantee the ISS and $\delta$ISS of deep GRUs.

\begin{proposition} \label{thm:deep:iss}
	A sufficient condition for the ISS of the deep GRU \eqref{eq:model:deepgru} is that each layer is ISS, i.e. it satisfies condition \eqref{eq:single:iss:condition}
	\begin{equation} \label{eq:deep:iss:condition}
		\| U_r^{i} \|_\infty  \, \bar{\sigma}_f^i < 1
	\end{equation}
	$\forall i \in \{1, ..., M\}$, where
	\begin{equation}
	    \bar{\sigma}_f^i = \sigma( \| W_f^i \quad U_f^i  \quad b_f^i \|_\infty ) .
	\end{equation}
\end{proposition}
\begin{proof}
    The deep GRU \eqref{eq:model:deepgru} can be considered as a cascade of ISS subsystems, and hence it is ISS, see \cite{jiang2001input}.
\end{proof}

\begin{proposition} \label{thm:deep:deltaiss}
	A sufficient condition for the $\delta$ISS of the deep GRU \eqref{eq:model:deepgru} is that each layer is $\delta$ISS, i.e. it satisfies \eqref{eq:single:deltaiss:condition}:
	\begin{equation} \label{eq:deep:deltaiss:condition}
		\| U_r^i \|_\infty \left( \frac{1}{4} {\widecheck{\lambda}}^i \| U_f^i \|_\infty +{\widecheck{\sigma}}_f^i \right) < 1 - \frac{1}{4} \frac{{\widecheck{\lambda}}^i +{\widecheck{\phi}}_r^i}{1 -{\widecheck{\sigma}}_z^i} \| U_z^i \|_\infty,
	\end{equation}
	$\forall i \in \{1, ..., M\}$, where, defining ${\widecheck{\lambda}}^0 = \max_{u \in \mathcal{U}} \| u \|_\infty = 1$,
	\begin{subequations} \label{eq:deep:deltaiss:sigmas}
	\begin{align}
		 {\widecheck{\sigma}}_z^i &= \sigma( \|  {\widecheck{\lambda}}^{i-1} W_z^i \quad {\widecheck{\lambda}}^i U_z^i \quad b_z^i \|_\infty ),\\
		 {\widecheck{\sigma}}_f^i &= \sigma( \|  {\widecheck{\lambda}}^{i-1} W_f^i \quad {\widecheck{\lambda}}^i U_f^i  \quad b_f^i \|_\infty ), \\
		 {\widecheck{\phi}}_r^i &= \phi( \| {\widecheck{\lambda}}^{i-1} W_r^i \quad {\widecheck{\lambda}}^i U_r^i \quad b_r^i \|_\infty ),
	\end{align}
	\end{subequations}
\end{proposition}
\begin{proof}
    See \ref{appendix:deep}.
\end{proof}

It is worth noting that, to the best of authors' knowledge, results guaranteeing the $\delta$ISS of a cascade of $\delta$ISS subsystems are only available for continuous-time systems \cite{angeli2002lyapunov}.
While an extension of this result to discrete-time systems may be object of further research efforts, in the proof of Proposition \ref{thm:deep:deltaiss} we opted to assess this property limited to the specific case of deep GRUs.
Moreover, it should be noted the conditions required by Proposition \ref{thm:deep:deltaiss} may be very conservative.
To relax the conservativeness of the approach, it is possible to assume that the GRU network is initialized inside the invariant set, i.e. $ {\check{\mathcal{X}}}= \mathcal{X}$, so that bounds \eqref{eq:deep:deltaiss:sigmas} can be eased and condition \eqref{eq:deep:deltaiss:condition} can be relaxed, as discussed in the following Corollary.

\begin{corollary} \label{cor:deep:deltaiss_relaxed}
    A sufficient condition for the $\delta$ISS of the deep GRU network \eqref{eq:model:deepgru}, initialized in $ {\check{\mathcal{X}}}= \mathcal{X}$, is that 
    \begin{equation} \label{eq:deep:deltaiss:condition_relaxed}
		\| U_r^i \|_\infty \left( \frac{1}{4} \| U_f^i \|_\infty + \bar{\sigma}_f^i \right) < 1 - \frac{1}{4} \frac{1 + \bar{\phi}_r^i}{1 - \bar{\sigma}_z^i} \| U_z^i \|_\infty,
	\end{equation}
	$\forall i \in \{1, ..., M\}$, where
	\begin{subequations} \label{eq:deep:deltaiss:sigmas_relaxed}
	\begin{align}
		\bar{\sigma}_z^i &= \sigma( \|  W_z^i \quad U_z^i \quad b_z^i \|_\infty )\\
		\bar{\sigma}_f^i &= \sigma( \|  W_f^i \quad U_f^i  \quad b_f^i \|_\infty ) \\
		\bar{\phi}_r^i &= \phi( \| W_r^i \quad U_r^i \quad b_r^i \|_\infty ).
	\end{align}
	\end{subequations}
\end{corollary}
\begin{proof}
    In light of Assumption \ref{ass:u_bounded}, ${\widecheck{\lambda}}^0 = \max_{u \in \mathcal{U}} \| u \|_\infty = 1$.
    Moreover, since $ {\check{\mathcal{X}}}= \mathcal{X}$, it follows that ${\widecheck{\lambda}}^i = 1$.
    Corollary \ref{cor:deep:deltaiss_relaxed} can thus be proven invoking Proposition \ref{thm:deep:deltaiss} with ${\widecheck{\lambda}}^i = 1$, for any $i \in \{ 1, ..., M\}$.
\end{proof}

\begin{remark} \label{rmk:weaker_deltaiss_deep}
    Note that the condition \eqref{eq:deep:deltaiss:condition_relaxed} involved by Corollary \ref{cor:deep:deltaiss_relaxed} is less conservative than the condition \eqref{eq:deep:deltaiss:condition} required by Proposition \ref{thm:deep:deltaiss}.
    Although Corollary \ref{cor:deep:deltaiss_relaxed} guarantees the $\delta$ISS only inside the invariant set $\mathcal{X}$, as discussed in Remark \ref{rmk:weaker_deltaiss_single}, it allows to state a similar but weaker stability-related property also if $\bar{x}_a \notin \mathcal{X}$ or $\bar{x}_b \notin \mathcal{X}$.
    Indeed, Lemma \ref{lemma:deep:invset} ensures that the state trajectories  exponentially converge into $\mathcal{X}$ in finite time, after which the $\delta$ISS property regularly applies.
\end{remark}
Finally, note that Proposition \ref{prop:iss_deltaiss} can be easily extended to deep GRUs to show that the $\delta$ISS condition \eqref{eq:deep:deltaiss:condition} implies the ISS condition \eqref{eq:deep:iss:condition}.

\section{Illustrative example} \label{sec:example}

\begin{figure}
	\centering
	\includegraphics[width=0.7 \linewidth]{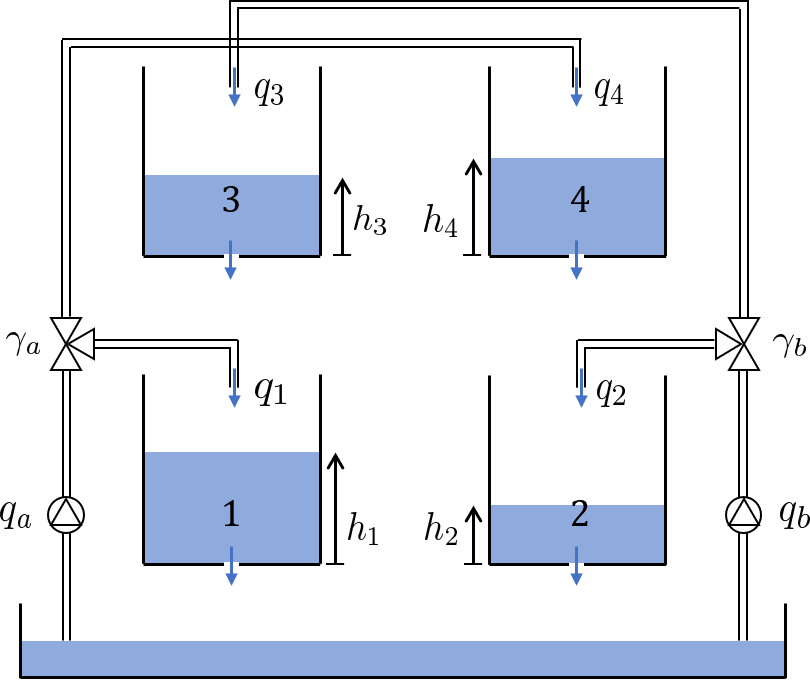}
	\caption{Scheme of the quadruple tank system.}
	\label{fig:example:qtscheme}
\end{figure}

\subsection{Benchmark system}
The proposed approach has been tested on the Quadruple Tank benchmark system reported in \cite{alvarado2011comparative}, with the goal of assessing the effects of the ISS and $\delta$ISS conditions on network's performances and training.
The system, depicted in Figure \ref{fig:example:qtscheme}, consists of four tanks containing water, with levels $h_1$, $h_2$, $h_3$, and $h_4$. 
Two controllable pumps supply the water flow rates $q_a$ and $q_b$ to the tanks.
The flow rate $q_a$ is split in $q_1$ and $q_4$ by a triple valve, so that $q_1 = \gamma_a q_a$ and $q_4 = (1 - \gamma_a) q_a$. Similarly, $q_b$ is split in $q_2$ and $q_3$, where $q_2 = \gamma_b q_b$ and $q_3 = (1 - \gamma_b) q_b$.
The system is hence characterized by the following equations \cite{alvarado2011comparative}:
\begin{equation}
\begin{aligned}
	\dot{h}_1 &= - \frac{a_1}{S} \sqrt{2gh_1} + \frac{a_3}{S} \sqrt{2gh_3} + \frac{\gamma_a}{S} q_a, \\
	\dot{h}_2 &= - \frac{a_2}{S} \sqrt{2gh_2} + \frac{a_4}{S} \sqrt{2gh_4} + \frac{\gamma_b}{S} q_b, \\
	\dot{h}_3 &= - \frac{a_3}{S} \sqrt{2gh_3} + \frac{1 - \gamma_b}{S} q_b, \\
	\dot{h}_4 &= - \frac{a_4}{S} \sqrt{2gh_4} + \frac{1 - \gamma_a}{S} q_a.
\end{aligned}
\end{equation}
The parameters of the system are reported in Table \ref{tab:example:parameters}. 
The control variables, expressed in $\frac{\text{m}^3}{\text{s}}$, are subject to saturation:
\begin{equation}
	q_a \in [ 0, 0.9 \cdot 10^{-3} ], \qquad q_b \in [ 0, 1.1 \cdot 10^{-3} ].
\end{equation}
The states are subject to physical constraints as well:
\begin{equation}
	h_{\{1, 2\}}\in [ 0, 1.36 ], \qquad h_{\{3, 4\}} \in [ 0, 1.3 ].
\end{equation}

\begin{table}
	\centering
	\caption{Benchmark system parameters}
	\label{tab:example:parameters}
	\begin{adjustbox}{max width=\columnwidth}
		\begin{tabular}{cccc|cccc}
		\toprule
		Parameter & Value & Units & $\,\,\,$ & $\,\,\,$ & Parameter & Value & Units \\
		\midrule 
		$a_1$ & $1.31 \cdot 10^{-4}$ & $\text{m}^2$  &&& $S$ & $0.06$ & $\text{m}^2$ \\
		$a_2$ & $1.51 \cdot 10^{-4}$ & $\text{m}^2$  &&& $\gamma_a$ & $0.3$ & \\
		$a_3$ & $9.27 \cdot 10^{-5}$ & $\text{m}^2$  &&& $\gamma_b$ & $0.4$ & \\
		$a_4$ & $8.82 \cdot 10^{-5}$ & $\text{m}^2$  &&& & & \\
		\bottomrule
		\end{tabular}
	\end{adjustbox}
\end{table}

It is assumed that only the levels $h_1$ and $h_2$ are measurable, i.e. the output of the system is $y = [h_1, h_2]^\prime$.
Therefore the system to be identified has two inputs, $u = [ q_a, q_b ]^\prime$, and two outputs.
Note that, to achieve proper results, the black-box model of this system should somehow implicitly model the two non-measurable states $h_3$ and $h_4$ and the states' saturation.

A simulator of the system has been implemented in Simulink, adding some white noise both to the inputs (standard deviation $5 \cdot 10^{-6}$) and to the measurements (standard deviation $0.005$).

{The ISS of the benchmark system can be easily verified, as it is an interconnection of sub-system, the tanks, which are ISS.
The $\delta$ISS of the system, instead, has been numerically assessed by means of Monte Carlo simulation campaign, in which the bounds to the $\delta$ISS functions $\beta_{\Delta}$ and $\gamma_{\Delta u}$ have been validated for $2000$ random initial conditions and input sequences.
Therefore, assuming that the system to identify is $\delta$ISS, in the following a model exhibiting the same property is trained. }

\subsection{Identification}
The simulator has been forced with Multilevel Pseudo-Ran\-dom Signals (MPRS) in order to properly excite the system, recording the input-output data with sampling time $\tau_s = 15 s$, so that enough data-points are collected in each transient.
The entire dataset is composed by $N_s = 30$ experiments, where each experiment $l \in \{ 1, ..., N_s \}$ is a collection of $T_s = 1500$ data-points $\big\{ u^{\{ l \}}(t)$, $y^{\{ l \}}(t) \big\}$.
The dataset has been split in a training set of $N_t = 20$ experiments, a validation set of $N_v = 5$ experiments, and a test set of {$N_{i} = 5$} experiment.
Moreover, in order to satisfy Assumption \ref{ass:u_bounded}, the dataset has been normalized so that for any data-point $u(t) \in [-1, 1]^2$ and $y(t) \in [-1, 1]^2$.

A deep GRU network with $M=3$ layers of $n_x = 7$ neurons has been implemented using TensorFlow 1.15 running on Python 3.7.
This network has been trained fulfilling the relaxed $\delta$ISS condition stated in Corollary \ref{cor:deep:deltaiss_relaxed}, so that it is guaranteed to be $\delta$ISS within its invariant set, as discussed in Remark \ref{rmk:weaker_deltaiss_deep}.
Since TensorFlow does not support constrained training, as discussed in Remark \ref{rmk:constraints} the constraint \eqref{eq:deep:deltaiss:condition_relaxed} must be relaxed. 
To do so, the following loss function is considered
\begin{equation} \label{eq:example:loss}
\begin{aligned}
    L =& \frac{1}{T_s - T_w} \sum_{k=T_w}^{T_s} \Big\| y(k, \bar{x}, \bm{u}^{\{l\}}) - y^{\{ l \}}(k) \Big\|^2 + \sum_{i=1}^{M} \rho(\nu^i),
\end{aligned}
\end{equation}
where $y(k, \bar{x}, \bm{u}^{\{l\}})$ denotes the open-loop prediction provided by the GRU \eqref{eq:model:deepgru}, initialized in the random state $\bar{x}$ and fed by the experiment's input sequence $\bm{u}^{\{l\}}$.
Therefore, the first part of the loss function $L$ is the prediction MSE associated to the $l$-th training sequence.
Specifically, MSE with a washout period $T_w=20$ is adopted, meaning that the prediction error in the first $T_w$ steps is not penalized, to accommodate the effects of the random initialization of the network \cite{bianchi2017overview}.
The second term of the loss function, $\rho(\nu^i)$, penalizes the violation of constraint \eqref{eq:deep:deltaiss:condition_relaxed} for each layer $i \in \{1, ..., M\}$. 
In particular, defining the constraint residual as
\begin{equation} \label{eq:example:residual}
	\nu^i = \| U_r^i \|_\infty \left( \frac{1}{4} \| U_f^i \|_\infty + \bar{\sigma}_f^i \right) - 1 + \frac{1}{4} \frac{1 + \bar{\phi}_r^i}{1 - \bar{\sigma}_z^i} \| U_z^i \|_\infty,
\end{equation}
where $\bar{\sigma}_f$, $\bar{\sigma}_z$, and $\bar{\phi}_r$ are defined as in \eqref{eq:deep:deltaiss:sigmas_relaxed}, it is evident that the constraint \eqref{eq:deep:deltaiss:condition_relaxed} is fulfilled if $\nu^i < 0$, otherwise it is violated. 
Denoting by $\varepsilon_{\nu} > 0$ the violation clearance,  $\rho(\nu^i)$ can be designed as a piece-wise linear cost, 
\begin{equation}
	\begin{aligned}
		\rho(\nu^i) = \rho\uss{+} \big[ \max ( \nu^i, -\varepsilon_{\nu} ) + \varepsilon_{\nu} \big] + \rho\uss{-} \big[ \min ( \nu^i, -\varepsilon_{\nu} ) + \varepsilon_{\nu} \big],
	\end{aligned}
\end{equation}
where $\rho\uss{+}$ and $\rho\uss{-} \ll \rho\uss{+}$ are hyperparameters that must to be tuned empirically. 
In this way $\nu^i$ is steered towards values smaller than $-\varepsilon_{\nu}$ while avoiding unnecessarily large residuals.
Furthermore, the weight $\rho\uss{+}$ should be sufficiently small to prioritize MSE's minimization.
In this example, we adopted $\rho\uss{+} = 2 \cdot 10^{-4}$,  $\rho\uss{-} = 2 \cdot 10^{-6}$, and a clearance $\varepsilon_\nu = 0.05$.

We carried out the training procedure using RMSProp as optimizer \cite{goodfellow2016deep}.
At each step of the training procedure, the loss function \eqref{eq:example:loss} is optimized for a batch given by a single sequence $l$. 
At the end of each training epoch, the training sequences are shuffled and, to avoid overfitting, an early-stopping rule is evaluated -- which halts the training when condition \eqref{eq:deep:deltaiss:condition_relaxed} is satisfied for all layers and the MSE on the validation dataset stops reducing.

\begin{figure}
	\centering
	\includegraphics[clip, width=\linewidth]{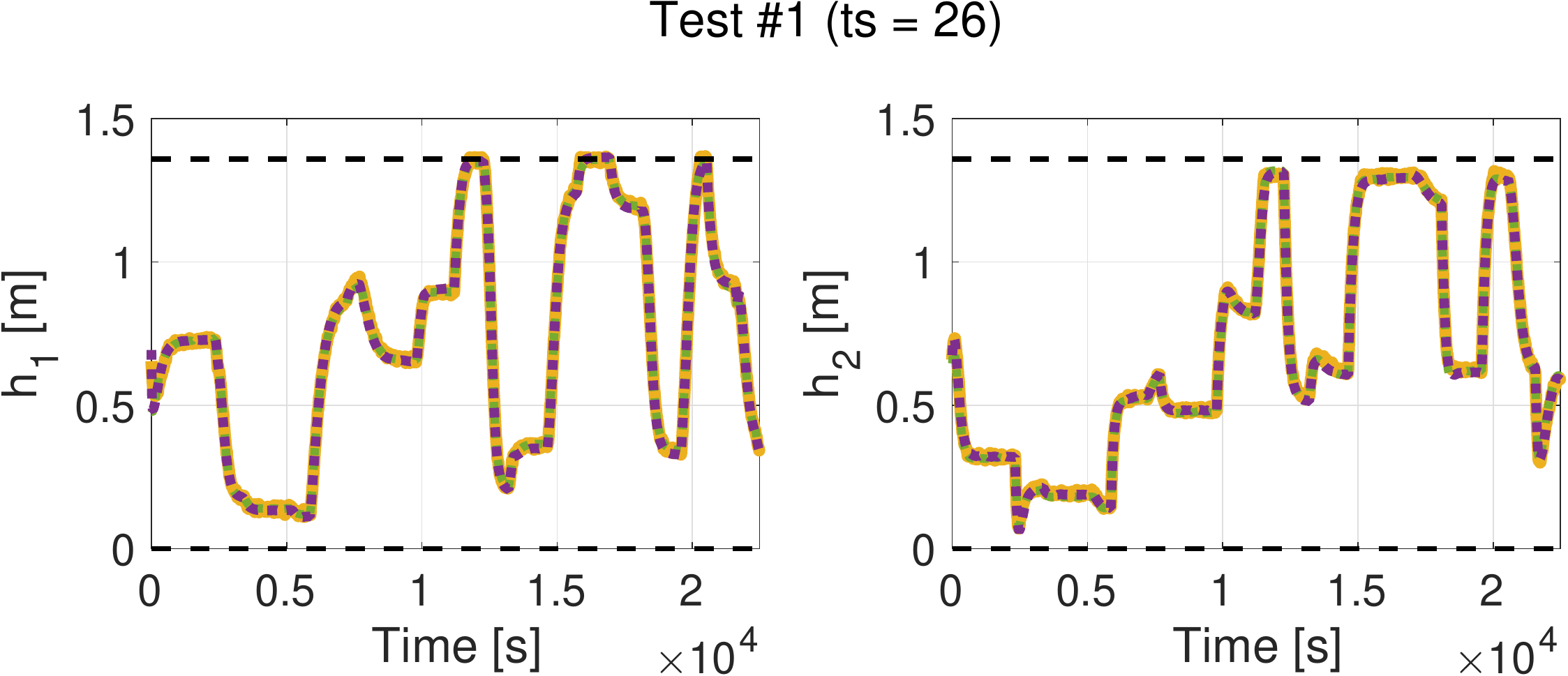} 
	\vspace{1mm}
	
	\includegraphics[clip, width=\linewidth]{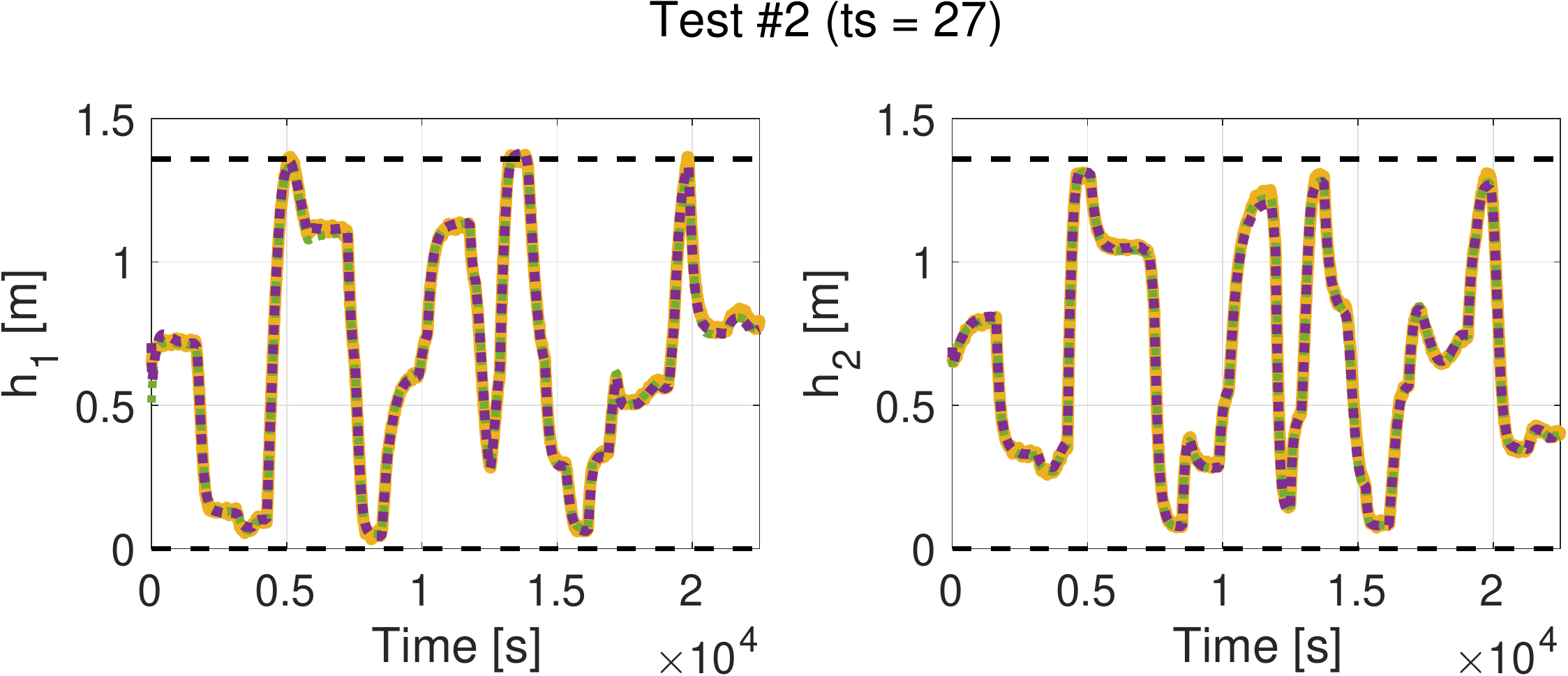}
	\vspace{1mm}
	
	\includegraphics[clip, width=\linewidth]{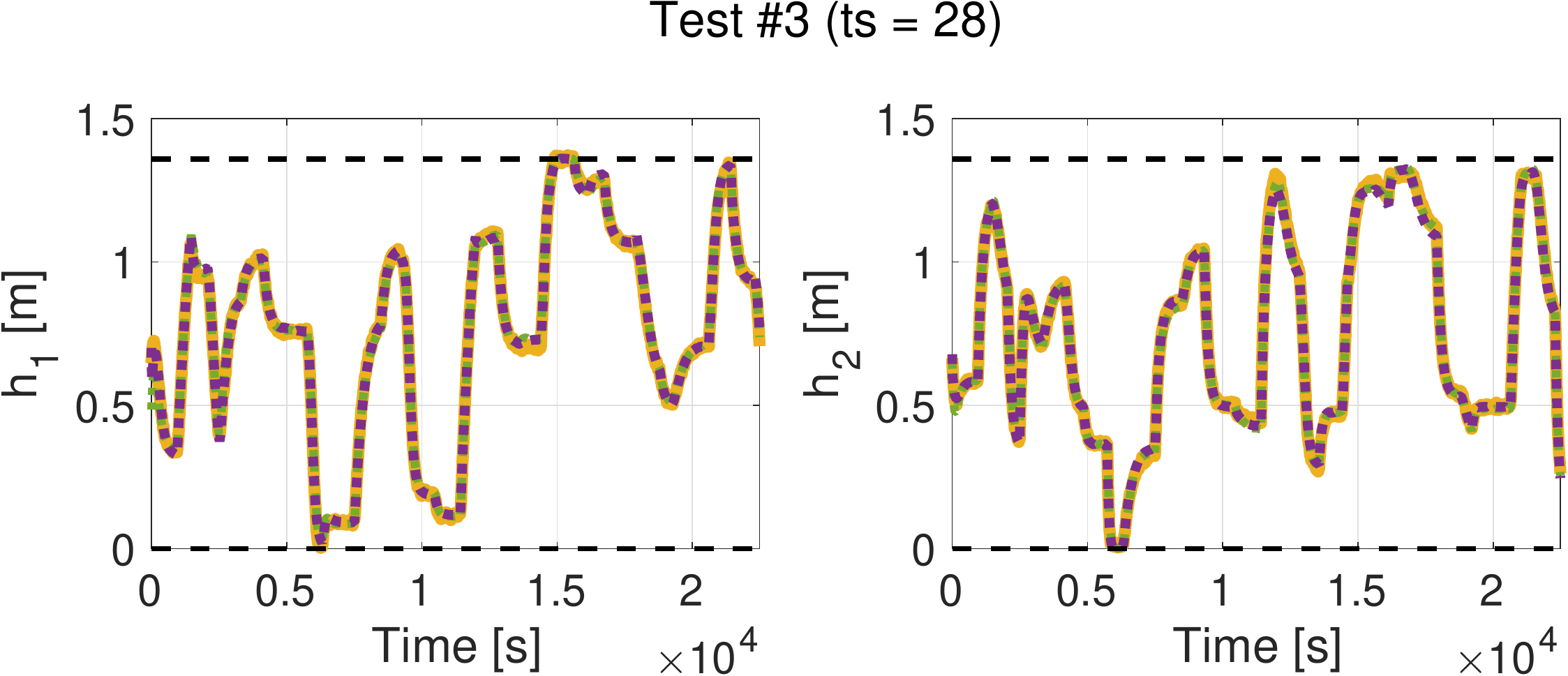}
	\vspace{1mm}
	
	\includegraphics[clip, width=\linewidth]{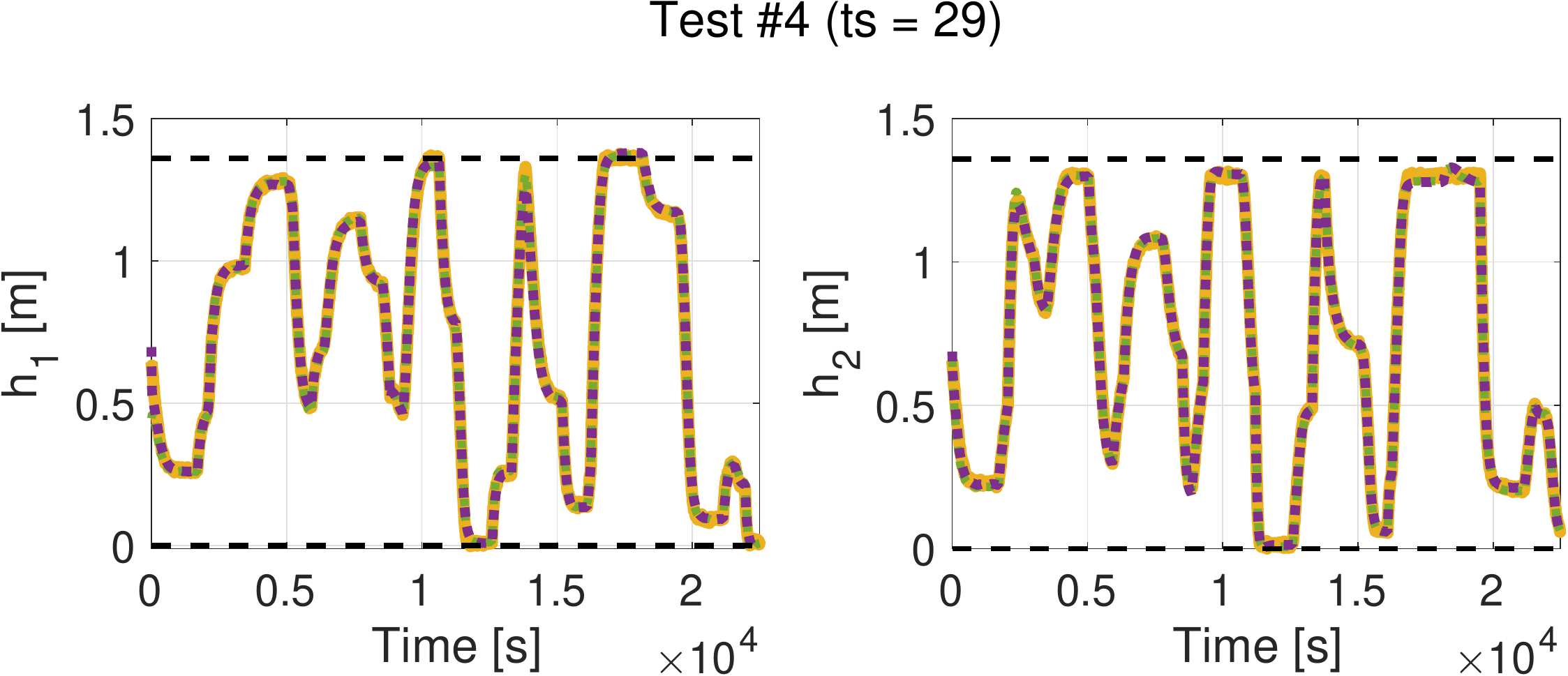}
	\vspace{1mm}
	
	\includegraphics[clip, width=\linewidth]{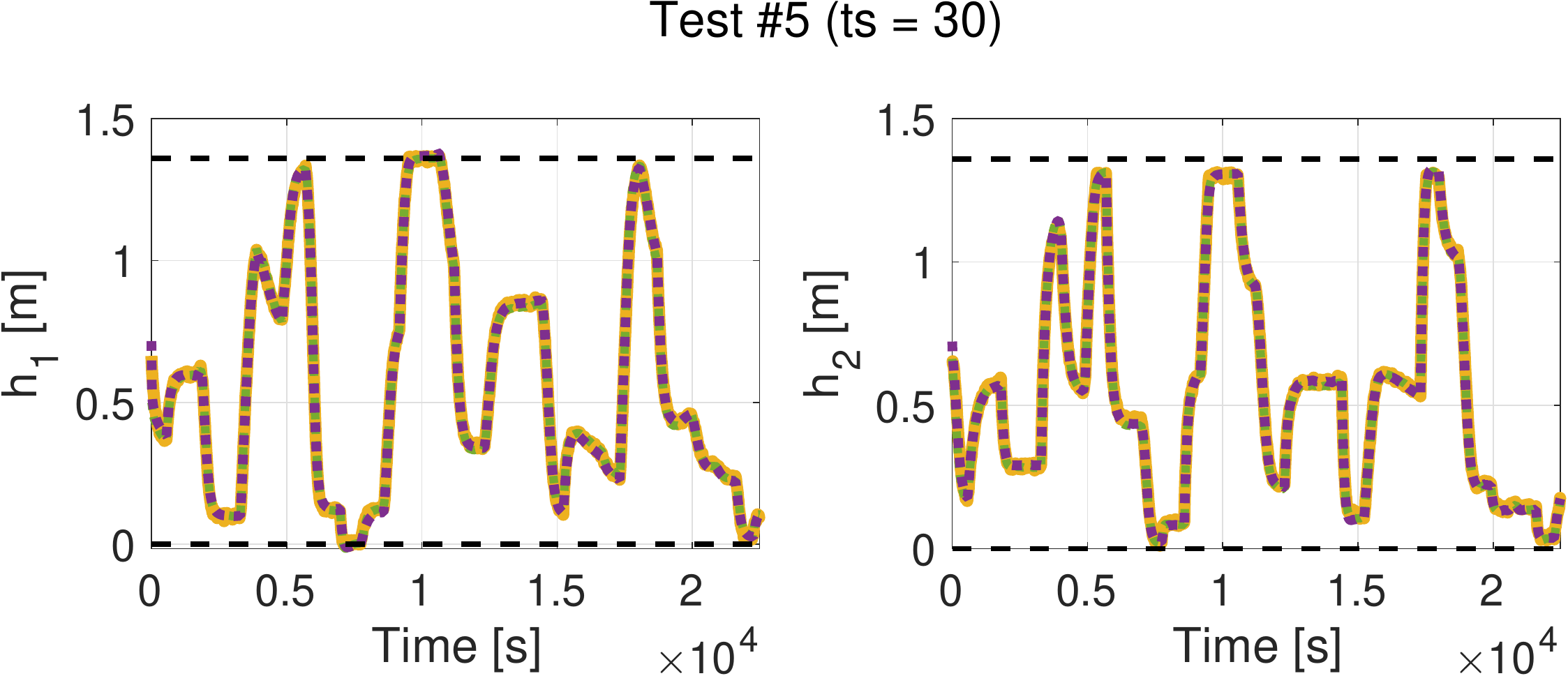}
	\vspace{1mm}
	
	\caption{Performances of the trained $\delta$ISS network (purple dotted line) on {the $N_i$ independent test sequences}, compared to the measured output (yellow solid line) and to a neural network with the same structure trained without the $\delta$ISS condition (green dotted line). Note that, due to the accuracy of the trained GRUs, the three lines are almost overlapping.}
	\label{fig:example:test}
\end{figure}

Finally, the  modeling performances of the trained network are assessed on  the {$N_i$ independent test sequences. The modeling performances are quantified on each test sequence $ts \in \{ N_s - N_i + 1, ..., N_s \} $} via the FIT index, defined as
\begin{equation}
	\text{FIT} = 100 \left( 1 - \frac{ \| \bm{y}\big(\bar{x}, \bm{u}^{\{ ts\}}\big) - \bm{y}^{\{ts\}} \| }{\| \bm{y}^{\{ts\}} - y_{avg} \|} \right),
\end{equation} 
where $\bm{y}(\bar{x}, \bm{u}^{\{ ts\}})$ denotes the open-loop simulation of the network, initialized in a random initial state $\bar{x}$ and fed by the test-set's input sequence $\bm{u}^{\{ts\}}$, while $y_{avg}$ is the mean of the test-set output sequence $\bm{y}^{\{ts\}}$. 

Figure \ref{fig:example:test} shows, {for each test sequence}, a comparison among the open-loop prediction by the trained $\delta$ISS network network, the real measured output, and the open-loop prediction by a network with the same structure but trained without enforcing the $\delta$ISS condition (i.e. with $\rho\uss{+} = \rho\uss{-} = 0$).
Remarkably, both the GRUs show impressive modeling capabilities, and they are able to model outputs' saturation.
Table \ref{tab:example:results} reports information about the training of the networks and the performances achieved. 
{The FIT index scored by the $\delta$ISS GRU network is, on average, $97.05 \%$ (min: $96.29 \%$, max: $97.5 \%$), while the FIT index scored by the unconstrained network is, on average, $97.21 \%$ (min: $96.58 \%$, max: $97.68 \%$).}

{It is worth highlighting that, unsurprisingly, the GRU trained without enforcing the stability condition fails to meet the $\delta$ISS sufficient condition, since its residuals $\nu^i$ are largely positive.}

{Finally, note that one may expect the conservativeness of the devised $\delta$ISS condition to come at the price of lower modeling performances, since \eqref{eq:deep:deltaiss:condition_relaxed} entails a bound on the maximum gains of the gates. 
However, results suggest that when the actual system generating the data itself enjoys the same stability property, the $\delta$ISS condition does not harm the modeling capabilities of the network. 
Nonetheless, as expected, the enforcement of such stability condition slows down the training of the network, meaning that more epochs are required to achieve the same modeling performances. 
This behavior is explained by the fact that, if the loss function's hyperparameters have been properly designed, at the beginning of the training procedure the order of the $\delta$ISS residual penalization term $\rho(\nu^i)$ is comparable to the order of the prediction's MSE. Hence, during the first epochs, the training algorithm trades modeling performances with the reduction of the $\delta$ISS residual.
}

\begin{table}
	\centering
	\begin{adjustbox}{max width=\columnwidth}
		\begin{tabular}{cccccc}
		    \toprule
		    \multirow{3}{*}{\shortstack[c]{$\delta$ISS \\ condition \\ enforced}} &
		    \multirow{3}{*}{\shortstack[c]{Average \\ FIT}}  &
		    \multirow{3}{*}{\shortstack[c]{Average \\ Validation\\MSE}} &
		    \multirow{3}{*}{\shortstack[c]{Average \\ Test\\MSE}} &
		    \multirow{3}{*}{\shortstack[c]{Train \\ Epochs}} &
		    \multirow{3}{*}{\shortstack[c]{Residuals \\ $\nu^i$}}  \\ \\ \\
			\midrule 
			Yes & $97.05\%$ & $3.4 \cdot 10^{-4}$ & $4.4 \cdot 10^{-4}$ & $640$ & $< -0.05$ \\
			No & $97.21\%$ & $3.2 \cdot 10^{-4}$ & $3.9 \cdot 10^{-4}$ & $383$ & $>3000$ \\
			\bottomrule
		\end{tabular}
	\end{adjustbox}
	\caption{Comparison of the results achieved by the GRUs trained with and without enforcing the $\delta$ISS condition.}
	\label{tab:example:results}
\end{table}

\section{Conclusions}
In this paper, sufficient conditions for the Input-to-State Stability (ISS) and Incremental Input-to-State stability ($\delta$ISS) of single-layer and deep Gated Recurrent Units (GRUs) have been devised, and guidelines on their implementation in a common training environment have been discussed.
{When GRUs are used to learn stable systems, the devised stability conditions allow to guarantee that the trained networks are ISS/$\delta$ISS, which is particularly useful during the synthesis of state observers and model predictive controllers.}
The proposed architecture has been tested on the Quadruple Tank benchmark system, showing remarkable modeling performances.
{Results suggest that, as long as the system to be identified is stable, the enforcement of stability conditions does not harm the modeling performances of GRUs, but it rather slows down the network's training procedure. }

{
\section*{Acknowledgments}
The authors would like to thank the Editor and the Reviewers for their valuable comments and suggestions.}

\bibliography{GRU_SCL_arxiv}
\appendix

\section{Proofs for single-layer GRUs} \label{appendix:single}
In the following, the proofs associated to the single-layer GRU model \eqref{eq:model:gru} are reported.
\bigskip 

\noindent \textbf{\emph{Proof of Lemma \ref{lemma:single:invset}}} \\ % \label{proof:lemma:single:invset_ini}
\noindent \emph{Case a.} $\bar{x} \in \mathcal{X}$
	
Applying Lemma \ref{lemma:single:invset_ini} iteratively it follows that  $x(k) \in \mathcal{X}$ for any $k \geq \bar{k} = 0$.
	
\medskip
\noindent \emph{Case b.} $\bar{x} \notin \mathcal{X}$, i.e. $\| \bar{x} \|_\infty > 1$ 

Consider the $j$-th component of \eqref{eq:model:proof_invset_ini:ltv}. 
In light of \eqref{eq:model:sigma_bounds}, at any time instant $\omega_j(k) \in (0, 1)$ and $\eta_j(k) \in (-1, 1)$.
More specifically, there exist $\ubar{\omega}(k)$, $\bar{\omega}(k)$, and $\ubar{\varepsilon}(k)$ such that $0 < \ubar{\omega}(k) \leq \omega_j(k) \leq \bar{\omega}(k) < 1$ and $ \lvert \eta_j(k) \lvert \leq 1 - \ubar{\varepsilon}(k) < 1$, for any component $j \in \{ 1, ..., n_x \}$.
	
These bounds are built as follows. 
By definition of infinity-norm, and since $\sigma$ is continuous and monotonically increasing, it follows that
\begin{equation*}
\begin{aligned}
    \omega_j(k) &\leq \max_{u} \left\| \sigma \left( \big[ W_z \quad U_z  \quad b_z \big] \, \begin{bmatrix}
		u \\ x(k) \\ 1_{n_x, 1}
		\end{bmatrix}\right) \right\|_\infty \\
	& \leq \max_{u}  \sigma \left( \big\| W_z \quad U_z \, \| x(k) \|_\infty \quad b_z \big\|_\infty \, \left\| \begin{array}{c}
		u \\ 1_{n_x, 1} \\ 1_{n_x, 1}
	\end{array} \right\|_\infty \right).
\end{aligned}
\end{equation*}
In light of Assumption \ref{ass:u_bounded}, the upper bound $\bar{\omega}_k$ can be  computed as
\begin{subequations} \label{eq:single:invset:bounds}
    \begin{equation}
        \bar{\omega}(k) = \sigma( \big\| W_z \quad U_z \, \| x(k) \|_\infty \quad b_z \big\|_\infty).
    \end{equation}
Moreover, owing to the symmetry of $\sigma$,
    \begin{equation}
        \ubar{\omega}(k) = \sigma ( -\big\| W_z \quad U_z  \, \| x(k) \|_\infty \quad b_z \big\|_\infty) = 1 - \bar{\omega}(k).
    \end{equation}
By similar arguments it is easy to show that
    \begin{equation}
	    \ubar{\varepsilon}(k) = 1 - \phi ( \big\| W_r \quad U_r \, \| x(k) \|_\infty   \quad b_r \big\|_\infty).
	\end{equation}
\end{subequations}
Taking the absolute value of \eqref{eq:model:proof_invset_ini:ltv}, one gets
\begin{equation} \label{eq:single:invset:int0}
    \lvert x_j(k+1) \lvert \, \leq \, \omega_j(k) \, \lvert x_j(k) \lvert + (1 - \omega_j(k)) \, \lvert \eta_j(k) \lvert.
\end{equation}
Then, if $\lvert x_j(k) \lvert \leq 1$, Lemma \ref{lemma:single:invset_ini} guarantees that $ \lvert x_j(\tilde{k}) \lvert \leq 1$, for any $\tilde{k} \geq k$.
If instead $\lvert x_j(k) \lvert > 1$, subtracting $\lvert x_j(k) \lvert$ from both sides of \eqref{eq:single:invset:int0}, we get
\begin{equation} \label{eq:single:invset:int1}
	\begin{aligned}
		\lvert x_j(k+1) \lvert - \lvert x_j(k) \lvert \leq -& (1 - \omega_j(k)) \lvert x_j(k) \lvert \\
	    &+ (1-\omega_j(k)) \, \lvert \eta_j(k) \lvert.
	\end{aligned}
\end{equation}
Then, {recalling that by definition $\lvert \eta_j(k) \lvert 1 - \ubar{\varepsilon}(k)$, noting that $\lvert x_j(k) \lvert > 1$ implies $-\lvert x_j(k) \lvert + 1 < 0$,} the following chain of inequalities holds
\begin{equation} \label{eq:single:invset:int2}
	\begin{aligned}
	    \lvert x_j(k+1) \lvert - \lvert x_j(k) \lvert &\leq  (1 - \omega_j(k)) \, \left( - \lvert x_j(k) \lvert + \lvert \eta_j(k) \lvert \right) \\
	    & < \ubar{\omega}(k) \, \left( - \lvert x_j(k) \lvert + 1 - \ubar{\varepsilon}(k) \right)  \\
	    & < - \ubar{\omega}(k) \, \ubar{\varepsilon}(k),
	\end{aligned}
\end{equation}
which entails that  $\lvert x_j(k+1) \lvert < \lvert x_j(k) \lvert$. 
Hence, repeating this argument for any state component, it follows that as long as $\| x(k) \|_\infty > 1$, $\| x(k+1) \|_\infty < \| x(k) \|_\infty$, which proves the first claim.
Note that, iterating this argument, it follows that 
{\begin{equation} \label{eq:single:invset:x_lessthan_x0}
    \| x(k) \|_\infty \leq \max( \| \bar{x} \|_\infty, 1),
\end{equation}
i.e. if $\bar{x} \notin \check{\mathcal{X}}$, the norm of the initial state is an upper bound of the norm of the state vector at any following time instant. 
Recalling that $\sigma$ and $\phi$ are monotonically increasing functions, \eqref{eq:single:invset:x_lessthan_x0} implies that the argument of the bounds stated in \eqref{eq:single:invset:bounds} is smaller than at the initial time instant.}
The bounds $\ubar{\omega}(0)$, $\bar{\omega}(0)$, and $\ubar{\varepsilon}(0)$ -- which can be easily computed, since the initial state $\bar{x}$ is known -- thus hold at any following time instant, i.e.
\begin{subequations} \label{eq:single:invset:bounds_k}
	\begin{gather}
		0 < \ubar{\omega}(0) \leq \ubar{\omega}(k) \leq \omega_j(k) \leq \bar{\omega}(k) \leq \bar{\omega}(0) < 1, \\
		\lvert \eta_j(k)  \lvert  \, \leq \, 1 -\ubar{\varepsilon}(k) \leq \, 1 - \ubar{\varepsilon}(0) < 1,
	\end{gather}
\end{subequations}
for any component $j \in \{ 1, ..., n_x \}$ and any instant $k \geq 0$.
Therefore, in light of \eqref{eq:single:invset:bounds_k}, \eqref{eq:single:invset:int2} implies that $x_j$ enters $[-1, 1]$ at most at the following time-step
\begin{equation}
    \bar{k}_j = \left\{ \begin{array}{lcl} 
    \left\lceil \frac{\lvert \bar{x}_j \lvert - 1}{\ubar{\omega}(0) \, \ubar{\varepsilon}(0)} \right\rceil & \quad & \text{if} \, \lvert \bar{x}_j  \lvert > 1 \\
    0  & \quad & \text{if} \, \lvert \bar{x}_j  \lvert \leq 1
    \end{array}
    \right. .
\end{equation}
The second claim is hence proven taking $\bar{k} = \max_j \bar{k}_j$.

Finally, we show that the convergence of each state component $x_j$ into its invariant set $[-1, 1]$ is exponential.
To this purpose, let us {write the evolution of the $j$-th state of system \eqref{eq:model:proof_invset_ini:ltv}} as $x_j(k) = x_{\alpha j}(k) + x_{\beta j}(k)$, where
\begin{subequations}
	\begin{align}
		x_{\alpha j}(k) &= \bigg( \prod_{t=0}^{k-1} \omega_j(t) \bigg) \bar{x}_j, \\
		x_{\beta j}(k) &= \sum_{t=0}^{k-1} \bigg( \prod_{h=t+1}^{k-1} \omega_j(h)\bigg) (1 - \omega_j(t)) \, \eta_j(t). \label{eq:model:proof_invset:motions:forced}
	\end{align}
\end{subequations}

\noindent Note that $x_{\alpha j}$ converges to zero. Indeed, in light of bounds \eqref{eq:single:invset:bounds_k}, 
\begin{equation}
    \lvert x_{\alpha j}(k) \lvert \leq \left\lvert \prod_{t=0}^{k-1} \omega_j(t) \right\lvert \, \lvert \bar{x}_j \lvert \, \leq \,  \big(\bar{\omega}(0)\big)^k \, \lvert \bar{x}_j \lvert
\end{equation}
tends to zero as $k \to \infty$. 
Concerning $x_{\beta j}$, taking the absolute value of \eqref{eq:model:proof_invset:motions:forced}, {and applying the bounds \eqref{eq:single:invset:bounds_k}}, one gets
\begin{equation}
    \begin{aligned}
        \lvert x_{\beta j}(k) \lvert & \leq \left\lvert \sum_{t=0}^{k-1} \left[ \prod_{h=t+1}^{k-1} \omega_j(h) - \prod_{h=t}^{k-1} \omega_j(h) \right]  \eta_j(t) \right\lvert \\
		& \leq  \left[ 1 - \prod_{t=0}^{k-1} \omega_j(t) \right] \, (1 - \ubar{\varepsilon}(0))  \\
		& \leq \big[ 1- \big( \ubar{\omega}(0) \big)^k \big] \, ( 1 - \ubar{\varepsilon}(0)).
    \end{aligned}
\end{equation}
By the triangular inequality $\lvert x_j(k) \lvert \leq \lvert x_{\alpha j}(k) \lvert + \lvert x_{\beta j}(k) \lvert$, which leads to
\begin{equation}
		\lvert x_j(k) \lvert \leq \big( \bar{\omega}(0) \big)^k \, \lvert \bar{x}_j \lvert + \Big[ 1 - \big(\ubar{\omega}(0)\big)^k \Big]  \, \left( 1 - \ubar{\varepsilon}(0) \right).
\end{equation}
This proves the third and last claim of Lemma \ref{lemma:single:invset}. $\hfill \blacksquare$
\bigskip

\noindent \textbf{\emph{Proof of Theorem \ref{thm:single:iss}}} \\
    \noindent\emph{Case a.} $\bar{x} \in {\check{\mathcal{X}}} = \mathcal{X}$ 

	Consider the $j$-th component of $x$, and let us remind that the time index $k$ is herein omitted for compactness. Let $z_j = z(u, x)_j$,  $f_j = f(u, x)_j$, and $r_j = r(u, x)_j$. Then from \eqref{eq:model:gru} it follows that
	\begin{equation*}
		x_j^+ = z_j \, x_j + (1 - z_j) \,  r_j.
	\end{equation*}
	Taking the absolute value, and since $z_j \in (0, 1)$, the previous equality becomes
	\begin{equation} \label{eq:single:iss:xj_abs}
		\lvert x_j^+ \lvert \leq  z_j \, \lvert x_j \lvert + (1 - z_j) \,  \left\lvert r_j \right\lvert.
	\end{equation}
	
	In light of Assumption \ref{ass:u_bounded} and Lemma \ref{lemma:single:invset},
	$\| u \|_\infty \leq 1$ and  $\| x \|_\infty \leq \| \bar{x} \|_\infty \leq 1$, since $\check{\mathcal{X}} = \mathcal{X}$ implies $\widecheck{\lambda} = 1$. 
	Hence the forget gate can be bounded as $\| f(u, x) \|_\infty\leq \bar{\sigma}_f$, where 
	\begin{equation} \label{eq:single:iss:sigma_f_bar}
		\begin{aligned}		
		\bar{\sigma}_f &= \max_{u, x} \left\| \sigma \left( \big[ W_f \quad U_f \quad b_f \big] \, \begin{bmatrix}
		u \\ x \\ 1_{n_x, 1}
		\end{bmatrix}\right) \right\|_\infty \\
		& = \max_{u, x}  \sigma \left( \left\| W_f \quad U_f \quad b_f \right\|_\infty \, \left\| \begin{array}{c}
			u \\ x \\ 1_{n_x, 1}
			\end{array} \right\|_\infty \right) \\
		& = \sigma \left( \left\| W_f \quad U_f \quad b_f \right\|_\infty \right). 
		\end{aligned}
	\end{equation}
	Analogously, it holds that $\| z(u, x) \|_\infty \leq \bar{\sigma}_z$, where
	\begin{equation} \label{eq:single:iss:sigma_z_bar}	
		\bar{\sigma}_z = \sigma \left( \left\| W_z \quad U_z \quad b_z \right\|_\infty \right).
	\end{equation}
	Recalling the definition of $r$ given in \eqref{eq:model:r_def}, owing to the Lipschitzianity and monotonicity of $\phi(\cdot)$, it holds that
	\begin{equation} \label{eq:single:iss:utilde_bound}
	\begin{aligned}
	\left\lvert r_j \right\lvert &  \leq \| r \|_\infty \leq \phi \left( \Big\| W_r \, u + U_r \, f \circ x + b_r \Big\|_\infty \right) \\
	& \leq \| W_r \|_\infty \, \| u \|_\infty + \| U_r \|_\infty \, \bar{\sigma}_f \, \| x \|_\infty + \| b_r \|_\infty
	\end{aligned}
	\end{equation}
	
	Noting that $\lvert x_j \lvert \leq \| x \|_\infty$, and applying \eqref{eq:single:iss:utilde_bound},  inequality \eqref{eq:single:iss:xj_abs} can thus be recast as	
	\begin{equation} \label{eq:single:iss:xj_abs_int}
	\begin{aligned}
		\lvert x_j^+ \lvert \, \leq \, & \Big[ z_j + (1 - z_j ) \| U_r \|_\infty \, \bar{\sigma}_f \Big] \, \| x \|_\infty \\
		& + (1 - z_j)  \| W_r \|_\infty \, \| u \|_\infty + (1 - z_j) \| b_r \|_\infty.
	\end{aligned}
	\end{equation}
	Since by definition $z_j \in [1 - \bar{\sigma}_z, \bar{\sigma}_z] \subseteq (0, 1)$ for any $j$, condition \eqref{eq:single:iss:condition} implies that there exists some $\delta \in (0, 1)$ such that
	\begin{equation*}
		z_j + (1 - z_j) \| U_r \|_\infty \, \bar{\sigma}_f \leq 1 - \delta,
	\end{equation*}
	allowing to re-write \eqref{eq:single:iss:xj_abs_int} as
	\begin{equation} \label{eq:single:iss:x_norm} 
		\| x^+ \|_\infty \leq ( 1 - \delta ) \| x \|_\infty + \bar{\sigma}_z \, \| W_r \|_\infty \| u \|_\infty + \bar{\sigma}_z \, \| b_r \|_\infty.
	\end{equation}
	Iterating \eqref{eq:single:iss:x_norm}, it is possible to derive that
	\begin{equation} \label{eq:single:iss:final} 
		\| x(k) \|_\infty \leq ( 1 - \delta )^k \| \bar{x} \|_\infty + \frac{\bar{\sigma}_z}{\delta} \, \| W_r \|_\infty \| \bm{u} \|_\infty + \frac{\bar{\sigma}_z}{\delta} \, \| b_r \|_\infty,
	\end{equation}
	{where the coefficients of $\| \bm{u} \|_\infty$ and $\| b_r \|_\infty$ have been majorized by the geometric series' limit $\sum_{t=0}^{k-1} (1-\delta)^t \leq \frac{1}{\delta}$.}
	
	Hence, system \eqref{eq:model:gru} is ISS with  $\beta(\| \bar{x} \|_\infty, k) = (1 - \delta)^k \| \bar{x} \|_\infty$, $\gamma_u(\| \bm{u} \|_\infty) = \frac{\bar{\sigma}_z}{\delta} \| W_r \|_\infty \| \bm{u} \|_\infty$, and $\gamma_b(\| b_r \|_\infty) = \frac{\bar{\sigma}_z}{\delta} \| b \|_\infty$.
	
	\medskip
	\noindent \emph{Case b.} $\bar{x} \in {\check{\mathcal{X}}} \supset \mathcal{X}$ 
	
    In light of Lemma \ref{lemma:single:invset}, the state trajectory $x$ converges into the invariant set $\mathcal{X}$ within a finite time instant $\bar{k}$.
    {Since for $k < \bar{k}$ the convergence into $\mathcal{X}$ is exponential and it is independent of the input sequence applied, for any $\delta \in (0, 1)$ there exists a sufficiently large $\mu > 0$ such that 
    \begin{equation} \label{eq:single:iss:before_kbar}
        \| x(k) \| \leq \mu (1 - \delta)^k \| \bar{x} \|_\infty.
    \end{equation}
    As soon as the state enters the invariant set, i.e. at $k = \bar{k}$, case \emph{a} applies. 
    Hence, for $k \geq \bar{k}$ it holds that
    \begin{equation}\label{eq:single:iss:after_kbar}
    \begin{aligned}
        \| x(k) \|_\infty \leq& ( 1 - \delta )^{k - \bar{k}} \| x(\bar{k}) \|_\infty + \frac{\bar{\sigma}_z}{\delta} \| W_r \|_\infty \| \bm{u} \|_\infty + \frac{\bar{\sigma}_z}{\delta} \, \| b_r \|_\infty.
    \end{aligned}
    \end{equation}
    Combining \eqref{eq:single:iss:before_kbar} and \eqref{eq:single:iss:after_kbar},}
    it follows that system \eqref{eq:model:gru} is ISS with functions $\beta(\| \bar{x} \|_\infty, k) = \mu (1 - \delta)^k \| \bar{x} \|_\infty$, $\gamma_u(\| \bm{u} \|_\infty) = \frac{\bar{\sigma}_z}{\delta} \| W_r \|_\infty \| \bm{u} \|_\infty$, and $\gamma_b(\| b_r \|_\infty) = \frac{\bar{\sigma}_z}{\delta} \| b \|_\infty$.  $\hfill\blacksquare$
    \bigskip
    
\noindent \textbf{\emph{Proof of Theorem \ref{thm:single:deltaiss}}} 
    Let us indicate by $x_{aj}$ and $x_{bj}$ the $j$-th components of $x_a = x(k, \bar{x}_{a}, \bm{u}_a, b_r)$ and $x_b = x(k, \bar{x}_{b}, \bm{u}_b, b_r)$, respectively.
    For compactness, we denote $f_a = f(x_a, u_a)$, $z_a = z(x_a, u_a)$, and $r_a = r(x_a, u_a)$, and we adopt the same notation for $f_b$, $z_b$, and $r_b$.
    Moreover, let $\Delta x = x_a - x_b$, and $\Delta u = u_a - u_b$.
	From \eqref{eq:model:gru} it holds that
	\begin{align*}
		\Delta x_j^+ &= z_{aj} x_{aj} + (1 - z_{aj}) r_{aj} - z_{bj} x_{bj} - (1 - z_{bj}) r_{bj}.
	\end{align*}
	Summing and subtracting the terms $z_{aj} x_{bj}$ and $(1 - z_{aj}) r_{bj}$, and taking the absolute value of $\Delta x_j^+$, it follows that
	\begin{equation} \label{eq:single:deltaiss:xj_abs}
	\begin{aligned}
		\lvert \Delta x^+_j \lvert &\leq  z_{aj} \lvert \Delta x_j \lvert + \lvert z_{aj} - z_{bj}\lvert \Big[ \lvert x_{bj} \lvert + \left\lvert r_{bj} \right\lvert \Big] \\
			& \quad\, + (1 - z_{aj}) \big\lvert r_{aj} - r_{bj} \big\lvert .
	\end{aligned}
	\end{equation}

	\begin{subequations} \label{eq:single:deltaiss:lipschitz}
	In light of Assumption \ref{ass:single:initial_state}, Lemma \ref{lemma:single:invset} entails that 
	\begin{equation} \label{eq:single:deltaiss:lipschitz:x_bound}
	    \lvert x_{\{a, b\} j} \lvert \, \leq  \| x_{\{a, b\}} \|_\infty \leq {\widecheck{\lambda}}.
	\end{equation}
	Thus, it follows that the forget gate can be bounded as
	\begin{equation}  \label{eq:single:deltaiss:lipschitz:f_bound}
	    \lvert f_{\{a, b\}j} \lvert \, \leq \| f_{\{a, b\}} \|_\infty \leq{\widecheck{\sigma}}_f,
	\end{equation}
	where $ {\widecheck{\sigma}}_f$ is defined as in \eqref{eq:single:deltaiss:bounds:f}.
	By similar arguments, since $ {\widecheck{\sigma}}_f < 1$,  the term $r_{bj}$ can be bounded as
	\begin{equation}
		\lvert r_{\{a,b\}j} \lvert \, \leq \| r_{\{a,b\}} \|_\infty \leq{\widecheck{\phi}}_r,
	\end{equation}
	where $ {\widecheck{\sigma}}_r$ is defined as \eqref{eq:single:deltaiss:bounds:r}.
	Then, owing to the Lipschitzianity of $\sigma$,
	\begin{equation}
	\begin{aligned}
		\lvert z_{aj} - z_{bj} \lvert & \leq \| z_a - z_b \|_\infty \\
		 & \leq \frac{1}{4} \big\| W_z (u_a - u_b) + U_z (x_a - x_b) \big\|_\infty \\
		 &\leq \frac{1}{4} \Big[ \| W_z \|_\infty \| \Delta u \|_\infty + \| U_z \|_\infty \| \Delta x \|_\infty \Big],
	\end{aligned}
	\end{equation}
	and
	\begin{equation}
	\begin{aligned}
	\| f_a - f_b \|_\infty \leq \frac{1}{4} \Big[ \| W_f \|_\infty \| \Delta u \|_\infty + \| U_f \|_\infty \| \Delta x \|_\infty \Big].
	\end{aligned}
	\end{equation}
	Exploiting the Lipschitzianity of $\phi$, {and in light of \eqref{eq:single:deltaiss:lipschitz:x_bound} and \eqref{eq:single:deltaiss:lipschitz:f_bound}}, the following chain of inequalities hold
	\begin{equation}
	\begin{aligned}
		  \big\lvert r_{aj} &- r_{bj} \big\lvert  \leq \big\| r_{a} - r_{b} \big\|_\infty \\
		 & \leq  \| W_r (u_a - u_b) + U_r (f_a \circ x_a - f_b \circ x_b) \|_\infty \\
		 & \leq \| W_r \|_\infty \| \Delta u \|_\infty + \\
		 & \qquad  + \| U_r \|_\infty \| (f_a - f_b) \circ x_a  + f_b \circ(x_a - x_b) \|_\infty \\
		 & \leq \| W_r \|_\infty \| \Delta u \|_\infty + \\
		 & \qquad  + \| U_r \|_\infty \Big( \| f_a - f_b \|_\infty \| x_a \|_\infty + \| f_b \|_\infty \| \Delta x \|_\infty \Big) \\
		 & \leq \| W_r \|_\infty \| \Delta u \|_\infty + \| U_r \|_\infty \Big( \frac{1}{4} {\widecheck{\lambda}} \| W_f \|_\infty \| \Delta u \|_\infty + \\
		 & \qquad + \frac{1}{4} {\widecheck{\lambda}} \| U_f \|_\infty \| \Delta x \|_\infty +{\widecheck{\sigma}}_f \| \Delta x \|_\infty  \Big).
	\end{aligned}
	\end{equation}
	\end{subequations}
	
	Combining \eqref{eq:single:deltaiss:xj_abs} and \eqref{eq:single:deltaiss:lipschitz} we thus obtain
	\begin{equation} \label{eq:single:deltaiss:contractive}
		\lvert \Delta x_j^+ \lvert \leq \alpha_{\Delta x} \| \Delta x \|_\infty + \alpha_{\Delta u} \| \Delta u \|_\infty,
	\end{equation}
	where
	\begin{equation} \label{eq:single:deltaiss:alphas}
	\begin{aligned}
		\alpha_{\Delta x} =& z_{aj} + \frac{1}{4}( {\widecheck{\lambda}} + {\widecheck{\phi}}_r) \| U_z \|_\infty  + \\
		& \qquad + (1 - z_{aj}) \| U_r \|_\infty \Big(  \frac{1}{4} {\widecheck{\lambda}} \| U_f \|_\infty  +{\widecheck{\sigma}}_f  \Big), \\
		\alpha_{\Delta u}  =& \frac{1}{4} ({\widecheck{\lambda}} +{\widecheck{\phi}}_r) \| W_z \|_\infty  + \\ 
		& \qquad + (1 - z_{aj})  \Big( \| W_r \|_\infty +   \frac{1}{4} {\widecheck{\lambda}} \| U_r \|_\infty \| W_f \|_\infty  \Big).
	\end{aligned}
	\end{equation}
	Condition \eqref{eq:single:deltaiss:condition} implies that there exists $\delta\sss{\Delta} \in(0, 1)$ such that $\alpha_{\Delta x} \leq 1 - \delta_{\Delta}$ for any $z_{aj}$.
	{Indeed, consider the inequality $\alpha_{\Delta x} < 1$, i.e.
	\begin{equation*}
	    \begin{aligned}
	        z_{aj} &+ \frac{1}{4}( {\widecheck{\lambda}} + {\widecheck{\phi}}_r) \| U_z \|_\infty \\
	        & \quad + (1 - z_{aj}) \| U_r \|_\infty \Big(  \frac{1}{4} {\widecheck{\lambda}} \| U_f \|_\infty  +{\widecheck{\sigma}}_f  \Big) < 1.
	    \end{aligned}
	\end{equation*}
	Bringing $z_{aj}$ on the right-hand side, and dividing by $1-z_{aj}$, which is surely positive, one obtains
	\begin{equation*}
	    \frac{1}{4} \frac{{\widecheck{\lambda}} + {\widecheck{\phi}}_r}{1 - z_{aj}} \| U_z \|_\infty + \| U_r \|_\infty \Big(  \frac{1}{4} {\widecheck{\lambda}} \| U_f \|_\infty  +{\widecheck{\sigma}}_f  \Big) < 1.
	\end{equation*}
	Moving the first term on the right-hand side, it follows that
	\begin{equation} \label{eq:single:deltaiss:ineq_zaj}
	    \| U_r \|_\infty \Big(  \frac{1}{4} {\widecheck{\lambda}} \| U_f \|_\infty  +{\widecheck{\sigma}}_f  \Big) < 1 - \frac{1}{4} \frac{{\widecheck{\lambda}} + {\widecheck{\phi}}_r}{1 - z_{aj}} \| U_z \|_\infty.
	\end{equation}
	Condition \eqref{eq:single:deltaiss:condition} entails that \eqref{eq:single:deltaiss:ineq_zaj} holds, since $z_{aj} \in [ 1-\widecheck{\sigma}_z, \widecheck{\sigma}_z ]$. }
	
	Therefore, as $\alpha_{\Delta x} \leq 1 - \delta_{\Delta}$, \eqref{eq:single:deltaiss:contractive} becomes
	\begin{equation} \label{eq:single:deltaiss:x_inf}
		\| \Delta x^+ \|_\infty \leq (1 - \delta_{\Delta}) \| \Delta x \|_\infty + \widecheck{\alpha}_{\Delta u} \| \Delta u \|_\infty,
	\end{equation}  
	where $\widecheck{\alpha}_{\Delta u}$ is the supremum of $\alpha_{\Delta u}$, computed replacing the minimum value of $z_{aj}$, i.e. $z_{aj} = 1 -{\widecheck{\sigma}}_z$, in \eqref{eq:single:deltaiss:alphas}.
	Iterating \eqref{eq:single:deltaiss:x_inf}, it is possible to derive that
	\begin{equation}
		\| \Delta x(k) \|_\infty \leq (1 - \delta_{\Delta})^k \| \bar{x}_{a} - \bar{x}_b \|_\infty + \frac{\widecheck{\alpha}_{\Delta u}}{\delta_{\Delta}} \| \bm{u}_a - \bm{u}_b \|_\infty,
	\end{equation}
	i.e. the system is $\delta$ISS with functions $\beta_{  \Delta}(\| \bar{x}_a - \bar{x}_b \|_\infty, k) = (1 - \delta\sss{\Delta})^k \| \bar{x}_a - \bar{x}_b \|_\infty$ and $\gamma_{{ \Delta} u}(\| \bm{u}_a - \bm{u}_b \|_\infty) = \frac{\widecheck{\alpha}_{ \Delta u}}{\delta\sss{\Delta}} \| \bm{u}_a - \bm{u}_b \|_\infty$. $\hfill\blacksquare$
	
\section{Proofs for deep GRUs} \label{appendix:deep}
In the following, the proofs associated to the deep GRU model \eqref{eq:model:deepgru} are reported.
\bigskip

\noindent\textbf{\emph{Proof of Lemma \ref{lemma:deep:invset}}} 
    Consider the first layer ($i=1$).
    Since $u^1 = u \in \mathcal{U}$, Lemma \ref{lemma:single:invset} can be straightforwardly applied to the first layer.
    Owing to the first claim of Lemma \ref{lemma:single:invset}, {the state of the first layer is bounded by $ {\widecheck{\lambda}}^1_0 = \max ( \| \bar{x}^1 \|_\infty, 1)$, see \eqref{eq:single:invset:x_lessthan_x0}.}
    
    {Thus, the input of the second layer ($i=2$), i.e. $u^2 = x^{1, +}$, is bounded by ${\widecheck{\lambda}}^1_0$.
    Lemma \ref{lemma:single:invset} can be then applied to the second layer, provided that the bounds \eqref{eq:single:invset:bounds} are suitably inflated to account the fact that the second layer's input, $u_2$, is no more unity-bounded. 
    In particular, letting
    \begin{equation}
        {\widecheck{\lambda}}_k^i = \max(\| x^i(k) \|_\infty, 1),
    \end{equation}
    since $\| x^{i-1}(k+1) \| \leq \widecheck{\lambda}_k^{i-1}$ and $\| x^{i}(k) \|_\infty \leq \widecheck{\lambda}_k^{i}$, the bounds can be computed as follows
    \begin{subequations}\label{eq:deep:invset:bounds}
    \begin{align}
        \bar{\omega}^{i}(k) =& \sigma \big( \big\| \, W_z^i \, {\widecheck{\lambda}}^{i-1}_k \quad  U_z^i \, {\widecheck{\lambda}}^{i}_k \quad b_z^i \big\|_\infty \big), \\
        \ubar{\omega}^i(k) =& 1 - \bar{\omega}^i(k), \\
        \ubar{\varepsilon}^i(k) =& 1 -  \phi \big( \big\| \, W_r^i \, {\widecheck{\lambda}}^{i-1}_k \quad  U_r^i \, {\widecheck{\lambda}}^{i}_k \quad b_r^i \big\|_\infty \big),
    \end{align}
    \end{subequations}}
    {so that $\ubar{\omega}^i(k) \leq \lvert \omega_j^i(k) \lvert \leq \bar{\omega}^i(k)$ and $\lvert \eta^i(k) \lvert \leq 1 - \ubar{\varepsilon}^i(k)$.
    Lemma \ref{lemma:single:invset} can then be applied to the second layer. }
    Iterating this arguments to all layers, the first and third claim of the Lemma are proved.
    
    {It is worth noticing that the first claim implies that the state of each layer $x^i(k)$ converges into its invariant set $\mathcal{X}^i$. 
    Hence, if $\bar{x}^i \notin \mathcal{X}^i$, $ \| x^i(k) \|_\infty$ contracts until $x^i$ enters the invariant set, which means that
    \begin{equation}
        x^i(k) \leq {\widecheck{\lambda}}_k^i \leq {\widecheck{\lambda}}_{k-1}^i \leq ... \leq {\widecheck{\lambda}}_0^i.
    \end{equation}
    Since $\sigma$ and $\phi$ are strictly monotonically functions and ${\widecheck{\lambda}}_k^i$ decreases with $k$, it holds that for any $k$
    \begin{subequations}
        \begin{gather}
            0 < \ubar{\omega}^i(0) \leq \ubar{\omega}^i_j(k) \leq \omega^i(k) \leq \bar{\omega}^i(k) \leq \ubar{\omega}^i(0) < 1, \\
            \lvert \eta_j^i(k) \lvert \leq 1 -  \ubar{\varepsilon}^i(k) \leq 1 - \ubar{\varepsilon}^i(0) < 1,
        \end{gather}
    \end{subequations}
    where $\bar{\omega}^i(0)$, $\ubar{\omega}^i(0)$, and $\ubar{\varepsilon}^i(0)$, are known quantities, as they depend only on the (known) initial states.}
    The second claim is hence verified by taking $\bar{k} = \max_i \bar{k}^i$, where $\bar{k}^i$ is the maximum convergence time of the $i$-th layer, computed as in the proof of Lemma \ref{lemma:single:invset}. $\hfill\blacksquare$
    \bigskip
    
\noindent\textbf{\emph{Proof of Proposition \ref{thm:deep:deltaiss}}}

    Let $\bar{x}_{a}^i \in {\check{\mathcal{X}}}^i$ and $\bar{x}_b^i \in {\check{\mathcal{X}}}^i$ be the pair of initial states for layer $i \in \{ 1, ..., M \}$, and let $\bm{u}_a$ and $\bm{u}_b$ be the pair of input sequences.
	We denote by $x_{a}^i = x^i(k, \bar{x}_a, \bm{u}_a, b_r)$ the state of network \eqref{eq:model:deepgru} initialized in $\bar{x}_{a}$ and fed by $\bm{u}_a$.
	The same notation is adopted for $x_{b}^i = x^i(k, \bar{x}_b, \bm{u}_b, b_r)$.
	For compactness, we denote $f_a^i = f(x_a^i, u_a^i)$, $z_a^i = z(x_a^i, u_a^i)$, and $r_a^i = r(x_a^i, u_a^i)$, and 
	 $f_b^i$, $z_b^i$, and $r_b^i$ are defined likewise.
	Finally, we denote $\Delta x = [ \Delta x^{1 \prime}, ..., \Delta x^{M \prime}]^{\prime}$, where $\Delta x^{i} = x^i_a -  x^i_b$, and $\Delta u = u_a - u_b$.
	
	First, let us point out that, in light of Assumption \ref{ass:deep:initial_state} and Lemma \ref{lemma:deep:invset}, at any time instant $\| x^i_{\{a, b\}} \|_\infty \leq \| \bar{x}^i_{\{a,b\}} \|_\infty \leq {\widecheck{\lambda}}^i$.
	Then, in light of \eqref{eq:model:deepgru:input}, it follows that the gates are bounded as $\| f_{\{a,b\}}^i \|_\infty \leq{\widecheck{\sigma}}_f^i$, $\| z_{\{a,b\}}^i \|_\infty \leq{\widecheck{\sigma}}_z^i$, and $\| r_{\{a,b\}}^i \|_\infty \leq{\widecheck{\phi}}_r^i$, where the bounds are defined as in \eqref{eq:deep:deltaiss:sigmas}.
	
	Thus, for any layer $i$, applying the same chain of inequalities as \eqref{eq:single:deltaiss:xj_abs}-\eqref{eq:single:deltaiss:x_inf}, one can derive that
	\begin{equation} \label{eq:deep:deltaiss:x_inf}
		\lvert \Delta x_j^{i, +} \lvert  \leq \alpha_{\Delta x}^i \| \Delta x^i \|_\infty + \alpha_{\Delta u}^i \| u_a^i - u_b^i \|_\infty,
	\end{equation}
	where $\alpha_{\Delta x}^i$ and $\alpha_{\Delta u}^i$ are defined as 
	\begin{equation} \label{eq:deep:deltaiss:alphas}
	\begin{aligned}
		\alpha_{\Delta x}^i =& z_{aj}^i + \frac{1}{4}( {\widecheck{\lambda}}^i + {\widecheck{\phi}}_r^i ) \| U_z^i \|_\infty  + \\
		& \qquad + (1  -  z_{aj}^i) \| U_r^i \|_\infty  \Big(  \frac{1}{4} {\widecheck{\lambda}}^i \| U_f^i \|_\infty + {\widecheck{\sigma}}_f^i \Big), \\
		\alpha_{\Delta u}^i =& \frac{1}{4} ({\widecheck{\lambda}}^i +{\widecheck{\phi}}_r^i) \| W_z^i \|_\infty + \\
		& \qquad + (1  - z_{aj}^i) \Big( \| W_r^i \|_\infty +  \frac{1}{4} {\widecheck{\lambda}}^i \| U_r^i \|_\infty \| W_f^i \|_\infty  \Big).
	\end{aligned}
	\end{equation}
	In light of condition \eqref{eq:deep:deltaiss:condition}, for any layer $i$ there exists $\delta_{\Delta}^i \in (0, 1)$ such that $\alpha_{\Delta x}^i < 1 - \delta_{\Delta}^i$. 
	Denoting by ${\widecheck{\alpha}}_{\Delta u}^i$, the supremum of $\alpha_{\Delta u}^i$ and applying \eqref{eq:model:deepgru:input}, \eqref{eq:deep:deltaiss:x_inf} can be recast as
	\begin{equation} \label{eq:deep:deltaiss:matrixineq}
		\begin{bmatrix}
		\| \Delta x^{1, +} \|_\infty \\
		\vdots \\
		\| \Delta x^{M, +} \|_\infty \\
		\end{bmatrix} \leq A_{M \Delta} \begin{bmatrix}
		\| \Delta x^{1} \|_\infty \\
		\vdots \\
		\| \Delta x^{M} \|_\infty \\
		\end{bmatrix} + B_{M{ \Delta}} \| \Delta u \|_\infty,
	\end{equation}
	where 
	\begin{subequations}
		\begin{equation}
		A_{M\Delta} = \begin{bmatrix}
		    1 - \delta^1_{\Delta} & 0 & ... & 0 \\
	    	(1-\delta^1_{\Delta}) \widecheck{\alpha}_{ \Delta u}^2  &  (1-\delta^2_{\Delta}) & ... & 0 \\
	    	\vdots & & \ddots & \vdots \\
	    	(1 - \delta^1_{\Delta}) \prod_{h=2}^{M} \widecheck{\alpha}_{\Delta u}^h   & ... & ... & (1 - \delta^M_{\Delta})
	    	\end{bmatrix}
	   	\end{equation}
	   	and
	   	\begin{equation}
	   	    B_{M \Delta} = \begin{bmatrix}
			{\widecheck{\alpha}}_{\Delta u}^1\\
			\vdots \\
			\prod_{h=1}^{M} {\widecheck{\alpha}}_{\Delta u}^h
		\end{bmatrix}.
		\end{equation}
	\end{subequations}
	Iterating \eqref{eq:deep:deltaiss:matrixineq}, and taking the norm of both sides, we get
	\begin{equation} \label{eq:deep:deltaiss:normineq}
	\begin{aligned}
	    \| \Delta x(k) \|_\infty \leq& \| A_{M \Delta}^k \|_\infty \| \Delta \bar{x} \|_\infty \\
		& \quad + \bigg\|  \sum_{t=0}^{k-1} A^{k-t-1}_{M \Delta} B_{M \Delta} \| \Delta \bm{u} \|_\infty \bigg\|_\infty.
	\end{aligned}
	\end{equation}
	Being the matrix $A_{M \Delta}$ triangular, its maximum eigenvalue is $\tilde{\lambda}_{\Delta} = \max_i (1 - \delta_{\Delta}^i) \in (0, 1)$ , which implies that it is Schur stable. 
	{Hence, there exists $\mu_{\Delta} > 0$ such that $\| A^k_{M \Delta} \|_\infty \leq \mu_\Delta \tilde{\lambda}_\Delta^k$ {(see \cite{terzi2021learning})}. Inequality \eqref{eq:deep:deltaiss:normineq} can be hence bounded as
	\begin{equation}
	    \begin{aligned}
	        \| \Delta x(k) \|_\infty \leq& \mu_\Delta \tilde{\lambda}_\Delta^k \| \Delta \bar{x} \|_\infty \\
	        & \quad +  \mu_\Delta \sum_{t=0}^{k-1} \tilde{\lambda}_\Delta^{k-t-1} \| B_{M \Delta} \|_\infty \, \| \Delta \bm{u} \|_\infty \\
	        \leq& \mu_\Delta \tilde{\lambda}_\Delta^k \| \Delta \bar{x} \|_\infty + \frac{\mu_\Delta}{1-\tilde{\lambda}_\Delta} \| B_{M \Delta} \|_\infty \, \| \Delta \bm{u} \|_\infty.
	    \end{aligned}
	\end{equation}}

	The GRU is $\delta$ISS with  $\beta_{ \Delta}( \| \Delta x \|_\infty, k) = \mu_{\Delta} \tilde{\lambda}_{\Delta}^k \| \Delta \bar{x} \|_\infty$ and $\gamma_{\Delta u}(\| \Delta \bm{u} \|_\infty) = \frac{\mu_\Delta}{1-\tilde{\lambda}_\Delta} \| B_{M \Delta} \|_\infty \| \Delta \bm{u} \|_\infty$.      $\hfill\blacksquare$
\end{document}